\documentclass{article}

\usepackage{algorithm}
\usepackage{algorithmicx}
\usepackage[noend]{algpseudocode}
\usepackage{amsfonts}
\usepackage{amsmath} 
\usepackage{amssymb}
\usepackage{amsthm}
\usepackage{authblk}
\usepackage[english]{babel}
\usepackage{balance} 
\usepackage{enumerate} 
\usepackage{epsfig}
\usepackage{epstopdf}
\usepackage{float}
\usepackage[T1]{fontenc}
\usepackage{graphicx}
\usepackage{hyperref}
\usepackage[utf8]{inputenc}
\usepackage{paralist}
\usepackage{subcaption}
\usepackage{tikz}
\usepackage{xspace}
\usepackage{mathtools}

\usetikzlibrary{shapes}

\newcommand{\tm}{\textsc{Throughput Maximization}~}

\newcommand{\cI}{\ensuremath{\mathcal{I}}}
\newcommand{\cJ}{\ensuremath{\mathcal{J}}}

\newcommand{\cS}{\ensuremath{\mathcal{S}}}
\newcommand{\cT}{\ensuremath{\mathcal{T}}}

\newcommand{\cW}{\ensuremath{\mathcal{W}}}

\newtheorem{theorem}{Theorem}
\newtheorem{lemma}{Lemma}

\newtheorem{corollary}{Corollary}
\newtheorem{definition}{Definition}

\newtheorem{observation}{Observation}

\newcommand{\eps}{{\varepsilon}}

\newcommand{\OPT}{\mbox{\sc OPT}}
\newcommand{\opt}{\rm{opt}}

\newcommand{\Op}{\ensuremath{\mathcal{O}}}

\title{Approximations for Throughput Maximization}
\author{Dylan Hyatt-Denesik \thanks{Department of Combinatorcs and Optimization, University of Waterloo.
Most of this work was done when the author was a graduate student in Department of 
Computing Science at U. of Alberta}
\quad
Mirmahdi Rahgoshay \thanks{Department of Computing Science, University of Alberta}
\quad Mohammad R. Salavatipour\thanks{Department of Computing Science, University of Alberta. Supported by NSERC.}
}

\date{}
\begin{document}

\maketitle

\begin{abstract}
In this paper we study the classical problem of throughput maximization. 
In this problem we have a collection $J$ of $n$ jobs, each having a release time $r_j$, deadline $d_j$,
and processing time $p_j$. They have to be scheduled non-preemptively on $m$ identical parallel machines.
The goal is to find a schedule which maximizes the number of jobs scheduled entirely in their $[r_j,d_j]$ window.
This problem has been studied extensively (even for the case of $m=1$). Several special cases of the problem
remain open. Bar-Noy et al. [STOC1999] presented an algorithm with ratio $1-1/(1+1/m)^m$ for $m$ machines, which
approaches $1-1/e$ as $m$ increases. For $m=1$, Chuzhoy-Ostrovsky-Rabani [FOCS2001] presented an algorithm with
approximation with ratio $1-\frac{1}{e}-\varepsilon$ (for any $\varepsilon>0$). Recently Im-Li-Moseley [IPCO2017] presented
an algorithm with ratio $1-1/e-\varepsilon_0$ for some absolute constant $\varepsilon_0>0$ for any fixed $m$. They also presented
an algorithm with ratio $1-O(\sqrt{\log m/m})-\varepsilon$ for general $m$ which approaches 1 as $m$ grows.
The approximability of the problem for $m=O(1)$ remains a major open question.
Even for the case of $m=1$ and $c=O(1)$ distinct processing times the problem is open (Sgall [ESA2012]).
In this paper we study the case of $m=O(1)$ and show that if there are $c$ distinct processing times, 
i.e. $p_j$'s come from a set of size $c$, then
there is a $(1-\eps)$-approximation that runs in time $O(n^{mc^7\eps^{-6}}\log T)$, where $T$ is the largest
deadline. Therefore, for constant $m$ and constant $c$ this yields a PTAS. 
Our algorithm is based on proving structural properties for a near optimum solution that allows
one to use a dynamic programming with pruning.
\end{abstract}

\newpage

\section{Introduction}
Scheduling problems have been studied in various fields, including Operations Research 
and Computer Science over the past several decades. However, there are still several fundamental problems that are not resolved.
In particular, for problems of scheduling of jobs with release times and deadlines in order to optimize some objective functions 
there are several problems left open (e.g. see \cite{SW99,Potts,S12}).
In this paper we consider the classical problem of throughput maximization. In this problem, we are given
 a set $J$ of $n$ jobs where each job $j \in J$ has a processing time $p_j$, a 
release time $r_j$, as well as a deadline $d_j$. The jobs are to be scheduled non-preemptively on a single 
(or more generally on $m$ identical) machine(s), which can process only one job at a time. 
The value of a schedule, also called its throughput,
is the number of jobs that are scheduled entirely within their release time and deadline interval.
Our goal is to find a schedule with maximum throughput. 

Throughput maximization is a central problem in scheduling that has been studied extensively in various settings (even
special cases of it are interesting open problems). They have numerous applications in 
practice \cite{FMT89,ARSU02,LZ98,HR98,YL97}.
The problem is known to be NP-hard (one of the list of problems in the classic book by Garey and Johnson \cite{GJ79}).
In fact, even special cases of throughput maximization have
attracted considerable attention. For the case of all $p_j$'s being equal in the weighted setting
(where each job has a weight and we want to maximize the total weight of scheduled jobs), 
the problem can be solved in polynomial time only when $m=O(1)$ (running time is exponential in $m$) \cite{BBKT,RS09}. 
The complexity of the problem is open for general $m$. For the case where all processing times are bounded by a constant
the complexity of the problem is listed as an open question \cite{S12}. It was shown in \cite{EW14} that even for $m=1$ and
$p_j\in\{p,q\}$ where $p$ and $q$ are strictly greater than $1$ the problem is NP-Complete.

\subsection{Related Works}
It appears the first approximation algorithms for this problem where given by Spieksma \cite{S98} 
where a simple greedy algorithm has shown to have approximation ratio $1/2$.
This algorithm will simply run the job with the least processing time between all
 the available jobs whenever a machine completes a job.
He also showed that the integrality gap of a natural Linear Program relaxation is $2$. Later on, Bar-Noy et al. \cite{BGNS01}
analyzed greedy algorithms for various settings and showed that for the case of $m$ identical machines
greedy algorithm has ratio $1-1/(1+1/m)^m$. This ratio is $1/2$ for $m=1$ and approaches $1-1/e$ as $m$ grows.

In a subsequent work, Chuzhoy et al. \cite{COR06} looked at a slightly different version, call it {\em discrete} version,
where for each job $j$, we are explicitly given a collection $\cI_j$ of intervals (possibly of different lengths) in which
job $j$ can be scheduled. A schedule is feasible if for each job $j$ in the schedule, $j$ is placed within one of the intervals
of $\cI_j$. This version (vs. the version defined earlier, which we call the ``continuous'' version) have similarities but none
implies the other. In particular, the discrete version can model the continuous version if one defines each interval of size
$p_j$ of $[r_j,d_j]$ as an interval in $\cI_j$. However, the number of intervals in $\cI_j$ defined this way can be as big as
$d_j-r_j+p_j$ which is not necessarily polynomial in input size. Chuzhoy et al. \cite{COR06} 
presented a $(1-1/e-\epsilon)$-approximation for the discrete version of the problem.
Spieksma \cite{S98} showed that the discrete version of the problem is $MAX$-$SNP$ hard using a 
reduction to a version of $MAX$-$3SAT$.
No such approximation hardness result has been proved for the continuous version.

Berman and DasGupta \cite{BD00} provided a better than $2$ approximation for the case when all the jobs are relatively 
big compared to their window size.
A pseudo-polynomial time exact algorithm for this case is presented by 
Chuzhoy et al. \cite{COR06} with running time $O(n^{poly(k)} T^4)$, where
 $k = max_j (d_j - r_j) / p_j$ and $T= max_j d_j$.

For the weighted version of the problem, 
\cite{B03} showed that when we have uniform processing time $p_j=p$, the problem is solvable in polynomial time for $m=1$. 
For $m=O(1)$ and with uniform processing time \cite{BBKT,RS09} presented polynomial time algorithms.
For general processing time $2$-approximation algorithms are provided in \cite{BD00,BBFNS01} and this ratio has been
the best known bound for the weighted version of the problem.
More recently, Im et al. \cite{ILM17} presented better approximations for throughput maximization for all values of $m$.
For the unweighted case, for some absolute $\alpha_0>1-1/e$, for any $m=O(1)$ and for any $\epsilon>0$ 
they presented an $(\alpha_0-\epsilon)$-approximation
in time $n^{O(m/\epsilon^5)}$. They also showed another algorithm with ratio $1-O(\sqrt{(\log m)/m}-\epsilon)$ (for any $\epsilon>0$)
on $m$ machines. This ratio approaches $1$ as $m$ grows. Furthermore, their $1-O(\sqrt{(\log m)/m}-\epsilon)$ ratio
extends to the weighted case if $T={\rm Poly}(n)$.

Bansal et al. \cite{BCKPSS07} looked at various scheduling problems and presented approximation algorithms with resource augmentation
(a survey of the many resource augmentation results in scheduling is presented in \cite{PST04}).
An $\alpha$-approximation with $\beta$-speed augmentation means a schedule in which the machines are $\beta$-times faster and
the total profit is $\alpha$ times the profit of an optimum solution on original speed machines.
In particular, for throughput maximization they presented a $24$-speed 1-approximation, i.e. a schedule with optimum throughput
however the schedule needs to be run on machines that are 24-times faster in order to meet the deadlines.
This was later improved by Im et al. \cite{ILMT15}, where they developed a dynamic programming framework for non-preemtive scheduling
problems. In particular for throughput maximization (in weighted setting) they present a quasi-polynomial time
$(1-\epsilon,1+\epsilon)$-bicriteria approximation (i.e. an algorithm that finds a $(1-\epsilon)$-approximate solution
using $(1+\epsilon)$ speed up in quasi-polytime). We should point out that the PTAS we present for $c$ distinct processing time
implies (as an easy corollary) a bicriteria QPTAS as well, i.e. a $(1-\epsilon)$-approximation using $(1+\epsilon)$-speed up.

For the problem of machine minimization, where we have to find the minimum number of machines with
 which we can schedule all the jobs, the algorithm provided in \cite{RT87} has approximation ratio
 $O(\sqrt{\log n/\log\log n})$ only when $OPT = \Omega(\sqrt{\log n/\log\log n})$,
and ratio $O(1)$ when $OPT = \Omega(\log n)$. 
Later Chuzhoy et al. \cite{CGKN04} presented an $O(OPT)$-approximation which is good for the instances with relatively small $OPT$.
Combining this with the earlier works implies an $O(\sqrt{\log n/\log\log n})$-approximation.
Chuzhoy and Naor \cite{CN06} showed a hardness of $\Omega(\log\log n)$ for the machine minimization problem.

Another interesting generalization of the problem is when we assign a height to each job as well and allow them to share the machine 
as long as the total height of all the jobs running on a machine at the same time is no more than $1$.
The first approximation algorithm for this generalization is provided by \cite{BBFNS01} which has ratio $5$. 
Chuzhoy et al. \cite{COR06}
improved it by providing an $(e-1)/(2e-1) > 0.3873$-approximation algorithm which is only working for the unweighted and 
discrete version of the problem.
The problem has also been considered in the online setting \cite{BKMMRRSW92, FN95, KS92, LT94}.

\subsection{Our Results}
Our main result is the following. 
Suppose that there are $c$ distinct processing times (i.e. each $p_j$ comes from a set of size $c$).

\begin{theorem}\label{theo:main}
For the throughput maximization problem with $m$ identical machines and 
$c$ distinct processing times for jobs, for any $\varepsilon>0$,
there is a $(1-\varepsilon)$-approximation algorithm that runs in time $n^{O(mc^7\varepsilon^{-6})}\log T$,
where $T$ is the largest deadline.
\end{theorem}

So for $m=O(1)$ and $c=O(1)$ we get a Polynomial Time Approximation Scheme (PTAS).
Note that even for the case of $m=1$ and $c=2$, the complexity of the problem has been listed as an open problem in \cite{S12},
however, it has been shown in \cite{EW14} that even for $m=1$ and
$p_j\in\{p,q\}$ where $p$ and $q$ are strictly greater than $1$ the problem is NP-Complete.
Our algorithm for Theorem \ref{theo:main} is obtained by proving some structural properties for near optimum solutions
and by describing a randomized hierarchical decomposition which allows us to do a dynamic programming.
In order to prove this we prove (and use at the base of our DP) the following (easier) special case:

\begin{theorem}\label{theo:base}
Suppose we are given $B$ intervals over the time-line where the machines are pre-occupied and cannot be used to run any jobs,
there are $R$ distinct release times, $D$ distinct deadlines, and $m$ machines, where $R,D,B,m\in O(1)$.
Then there is a PTAS for throughput maximization with time $2^{\eps^{-1}\log^{-4}(1/\eps)}+{\rm Poly}(n)$.
\end{theorem}

An easy corollary of Theorem \ref{theo:main} is the following.
If the largest processing time $p_{max}=\rm{Poly}(n)$ then we get a quasi-polynomial time $(1-\epsilon)$-approximation
using $(1+\epsilon)$-speed up of machines. This result of course was already obtained in \cite{ILMT15}.



\section{Preliminaries}
Recall that we have a set $J$ of $n$ jobs where each job $j \in J$ has a processing time $p_j$, a 
release time $r_j$ as well as a deadline $d_j$, we assume all these are integers in the range $[0,T]$
(we can think of $T$ as the largest deadline).
The jobs are to be scheduled non-preemptively on $m$ machines which can process only one job at a time. 
We point out that we do not require $T$ to be poly-bounded in $n$. 
For each job $j \in J$ we refer to $[r_j, d_j]$ as span of job $j$, denoted by $span_j$.
We use $\OPT$ to denote an optimum schedule and $\opt$ the value of it.
In the weighted case, each job $j$ has a weight/profit $w_j$ which we receive if we schedule the job within its span.
The goal in throughput maximization is to find a feasible schedule with maximum weight of jobs.
Like most of the previous works, we focus on the unit weight setting (so our goal is to find a schedule with maximum 
number of jobs scheduled).


We also assume that for each $p \in P$, all the jobs with processing time $p$ in an optimum solution
are scheduled based on earliest deadline first rule; 
 which says that at any time when there are two jobs with the same processing time available the one with the earliest deadline 
would be scheduled. This is known as Jackson rules and we critically use it in our algorithms.

{\bf Outline:} We start by presenting the proof of Theorem \ref{theo:main}. 
We defer the proof of Theorem \ref{theo:base} and several details of proof of Theorem \ref{theo:main}
to Sections \ref{sec:details} and \ref{sec:theobase}.

\section{Proof of Theorem \ref{theo:main}}
In this section we prove Theorem \ref{theo:main}. For ease of exposition, we present the proof for the case of $m=1$ machine
only and then extend it to the setting of multiple machines.

\subsection{Overview of the Algorithm}
At a high level, the algorithm removes a number of jobs so that there is a structured near optimum solution.
We show that the new instance has some structural properties that is amenable to a dynamic programming.
At the lowest level of dynamic programming we have disjoint instances of the problem,
each of which has a set of jobs with only a constant size set of release times
and deadlines, with possibly a constant number of intervals of time being blocked from being used.
For this setting we use the algorithm of Theorem \ref{theo:base}.
We start (at level zero) by breaking the interval $[0,T]$ into a constant $q$ (where $q$ will be dependent on $\varepsilon$) 
number of (almost) equal size intervals, with a random offset. 
Let us call these intervals $a_{0,1},a_{0,2},\ldots,a_{0,q}$. Assume each
interval has size exactly $T/q$, except possibly the first and last (and for simplicity assume $T$ is a power of $q$).
For jobs whose span is relatively large, i.e. spans at least $\lambda$ (where $\frac{1}{\varepsilon} \leq \lambda \leq \varepsilon q$) 
intervals, while their processing time is relatively small (much smaller than $T/q$), based on the random choice 
of break points for the intervals, 
we can assume the probability that the jobs position in the optimum solution is intersecting two intervals is very small. Hence,
ignoring those jobs (at a small loss of optimum), we can assume that each of those jobs are scheduled (in a near optimum solution)
entirely within one interval. For each of them
we ``guess'' which of the $\lambda$ intervals is the interval in which they are scheduled and pass down
the job to an instance defined on that interval. For jobs whose span is very small (fits entirely within one interval), 
the random choice of the $q$ intervals, implies that the probability of their span being ``cut'' by these intervals is 
very small (and again we can ignore those that have
been cut by these break down). For medium size spans, we have to defer the decision making for a few iterations.
We then try to solve each of the $q$ instances, independently and recursively; i.e. we break the intervals again into roughly $q$
equal size intervals and so on. If and when an instance generated has only $O(1)$ release times or deadlines we stop the recursion
and use the algorithm of Theorem \ref{theo:base} to find a near optimum solution. So considering the hierarchical structure of this
recursion, we have a tree with at most $O(\log_q T)$ depth and at most $O(n)$ leaves, which is polynomial in input size.
There are several technical details that one needs to overcome in this paradigm. One particular technical difficulty is for some
jobs we decide to re-define their span to be a smaller subset of their original span by increasing their release time a little
and decreasing their deadline a little. We call this procedure, cutting their ``head'' and ``tail''. 
This will be a key property in making our algorithm work. We will show that under some
moderate conditions, the resulting instance still has a near optimum solution. This allows us to reduce the number of guesses we
have to make in our dynamic program table and hence obtain Theorem \ref{theo:main}. 
We should point out that the idea of changing the span or start/finish of a job was done in earlier works. However, using speed-up
of machines one could ``catch up'' in a modified schedule with a near optimum one. The difficulty in our case is we do not have
machine speed up.

\subsection{Structure of a Near Optimum Solution}
Consider an optimum solution $\OPT$. One observation we use frequently is that such a solution is left-shifted, meaning that
the start time of any job is either its release time or the finish time of another job.
Therefore, we can partition the jobs in schedule $\OPT$ into continuous segments of jobs being run
whose leftmost points are release times and the jobs in each segment are being run back to back. 
We call the set of possible rightmost points of these segments ``slack times''.

\begin{definition} (Slack times). 
Let slack times $\Psi$ be the set of points $t$ such that there is
a release time $r_i$ and a (possibly empty) subset of jobs $J'\subseteq J$, such that $t = r_i + \sum_{j\in J'} p_j$
\end{definition}

So the start time and finish time of each job in an optimum solution is a slack time.
The following (simple) lemma bounds the size of $\Psi$

\begin{lemma}\label{lem:slack}
There are at most $n^{c+1}$ different possible slack times, where $c$ is the number of distinct processing times.
\end{lemma}
\begin{proof}
We upper bound number of distinct $r_i + \sum_{j\in J'} p_j$ values. First note that there are only $n$ different
$r_i$ values. Also, for each set $J'\subseteq J$, the sum $\sum_{j\in J'} p_j$ can have
at most $n^c$ possible values as the number of jobs in $J'$ with a specific processing time can be at most $n$ and
we assumed there are only $c$ distinct processing times.
\end{proof}

Given error parameter $\varepsilon > 0$ we set $q = 1/\varepsilon^2$, $k = \log_q T$ 
and for simplicity of presentation suppose $T$ is a power of $q$. We define a hierarchical set of partitions on 
interval $[0, T]$. For each $0 \leq i \leq k$, $I_i$ is a partition of $[0, T]$ into $q^{i+1}+1$ many intervals 
such that, except the first and the
 last intervals, all have length $\ell_i=T/{q^{i+1}}$, and the sum of the sizes of the
first and last interval is equal to $\ell_i$ as well.
We choose a universal random offset for the start point of the first interval.
More precisely, we pick a random number $r_0 \in [0, \frac{T}{q}]$ and interval $[0, T]$ is partitioned into $q+1$ intervals
 $I_0=\{a_{0,0}, a_{0,1}, ..., a_{0,q}\}$, where $a_{0,0} = [0, r_0]$, and
 $a_{0,t} = [(t-1)\frac{T}{q}+r_0, t\frac{T}{q}+r_0]$ for $1 \leq t \leq q-1$ and $a_{0,q} = [T - \frac{T}{q} + r_0, T]$. Note
 that the length of all intervals in $I_0$ is $\frac{T}{q}$, except the first and the last which have their length randomly chosen
and the sum of their lengths is $\frac{T}{q}$.

Similarly each interval in $I_0$ will be partitioned into $q$ many intervals to form partition $I_1$ with 
each interval in $I_1$ having length $\frac{T}{q^2}$ except the first interval obtained from breaking $a_{0,0}$ and
 the last interval in $I_1$ obtained from breaking $a_{0,q}$, which may be partitioned into less than $q$ many, 
based on their lengths. All intervals 
in $I_1$ have size $\frac{T}{q^2}$ except the very first one and the very last one. We do this iteratively and break intervals
of $I_i$ (for each $i\geq 0$) into $q$ equal sized intervals to obtain $I_{i+1}$ (with the exception of the very first and the
very last interval of $I_{i+1}$ might have lengths smaller).

We set $\lambda = 1/\varepsilon = \varepsilon q$ and partition the jobs into classes $\cJ_0, \cJ_1, ..., \cJ_k, \cJ_{k+1}$, based
 on the size of their span. For each $1 \leq i \leq k$, job $j \in \cJ_i$ if 
$\lambda \cdot \ell_i \leq \lvert span_j \rvert < \lambda \cdot \ell_{i-1}$. Also $j \in \cJ_0$ (and $j \in \cJ_{k+1}$) if
$\lambda \ell_0 \leq \lvert span_j \rvert$ (and $\lvert span_j \rvert < \lambda \cdot \ell_k$).
For each interval $a_{i,t}$ in level $I_i$, we denote the set of jobs whose span is entirely inside $a_{i,t}$ by $J(a_{i,t})$.

Based on our definitions of interval levels and job classes, we can say that for each $0 \leq i \leq k$ if $j \in \cJ_i$,
 then $span_j$ would have intersection with at most $\lambda+1$ (or fully spans at most $\lambda - 1$) many consecutive 
intervals from $I_{i-1}$ and at least $\lambda$ many consecutive intervals from $I_i$. Suppose $j \in I_i$ and $span_j$ has
 intersection with $a_{i, t_j}, a_{i, t_j+1}, ..., a_{i, t'_j}$ from $I_i$, then define $span_j \cap a_{i, t_j}$ and
 $span_j \cap a_{i, t'_j}$ as $head_j$ and $tail_j$, respectively.

We consider two classes of jobs as ``bad'' jobs and show that there is a near optimum solution without any bad jobs.
The first class of bad jobs are those that we call ``span-crossing''.
For each job $j \in J$, we call it "span-crossing" if $j \in \cJ_i$ for some $2 \leq i \leq k+1$
(so $\lambda \cdot \ell_i \leq \lvert span_j \rvert < \lambda \cdot \ell_{i-1}$), and its span has intersection
 with more than one interval in $I_{i-2}$.

\begin{lemma} \label{lmm:span_crossing}
Based on the random choice of $r_0$ (while defining intervals), the expected number of span-crossing jobs in
the optimum solution is at most $\frac{\lambda+1}{q}\opt=O(\varepsilon\opt)$.
\end{lemma}

\begin{proof}
Observe that because $j \in \cJ_i$, we have $\lvert span_j \rvert < \lambda \cdot \ell_{i-1}$. This means that the $span_j$ would 
have intersection with at most $\lambda+1$ (or fully spans at most $\lambda - 1$) many consecutive intervals from $I_{i-1}$. 
Also because of the random offset while defining $I_0$, and since $\ell_{i-2}=q\cdot \ell_{i-1}$,
the probability that job $j$ being "span-crossing" will be 
at most $\frac{\lambda+1}{q}$. 
\end{proof}

So, we can assume with sufficiently high probability, that there is a $(1-O(\varepsilon))$-approximate solution
 with no span-crossing jobs. 
The second group of bad jobs are defined based on their processing time and their position in the optimum solution. We then prove 
that by removing these type of jobs, the profit of the optimum solution will be decreased by a small factor. 
For each job $j \in J$, we call it "position-crossing" if $\ell_i \leq p_j < \ell_{i-1}$ for some 
$2 \leq i \leq k+1$, and its position in $\OPT$ has intersection with more than one interval in $I_{i-2}$.

\begin{lemma} \label{lmm:span_optimum_crossing}
The expected number of position-crossing jobs in $\OPT$ is at most $\frac{1}{q}\opt=O(\varepsilon^2\opt)$. 
\end{lemma}

\begin{proof}
Consider $\OPT$ and suppose that $j \in J$ is a job with $\ell_i \leq p_j < \ell_{i-1}$.
Observe that $j$ can have intersection with at most $2$ intervals in $I_{i-2}$ because of its size. 
Considering our random offset to define interval levels, the probability of job $j$ 
being a position-crossing (with respect to the random intervals defined)
 would be at most $\frac{1}{q}$ (since $p_j<\ell_{i-1}=\frac{\ell_{i-2}}{q}$).
Thus, the expected number of position-crossing jobs in $\OPT$ is at most $\opt/q$. 
\end{proof}

Hence, using Lemmas \ref{lmm:span_crossing} and \ref{lmm:span_optimum_crossing}, with sufficiently high probability,
there is a solution of value at least $(1-O(\varepsilon))\opt$ without any span-crossing or position-crossing jobs.
We call such a solution a canonical solution.

From now on, we suppose the original instance $\cI$ is changed to $\cI'$ after we first defined the intervals 
randomly and removed all the span-crossing jobs. So we focus (from now on) on finding a near optimum
feasible solution to $\cI'$ that has no
position-crossing jobs. By $\OPT'$ we mean such a solution of maximum value for $\cI'$; we call that a canonical 
optimum solution.
If we find a $(1-O(\varepsilon))$-approximation to $\OPT'$ (that has no position-crossing jobs), then
 using the above two lemmas we have a
$(1-O(\varepsilon))$-approximate solution to $\cI$. So with $\OPT'$ being an optimum solution to $\cI'$ with no 
position-crossing jobs we let $\opt'$ be its value.

\subsection{Finding a Near Optimum Canonical Solution }
As a starting point and warm-up, we consider the special case where 
instance $\cI'$ only consists of jobs whose processing time is relatively big compared to their
span and show how the problem could be solved. Consider the extreme case where
for each $j\in J$, $p_j = \lvert span_j \rvert$. 
In this case the problem will be equivalent to the problem of finding a maximum
 independent set in an interval graphs which is solvable in polynomial time \cite{G04}.
The following theorem shows that if $p_j\geq \frac{\lvert span_j \rvert}{\lambda}$ for each $j\in J$ (which 
we call them ``tight'' jobs),
 then we can find a good approximation as well. Therefore, it is the ``loose'' jobs (those whose processing time $p_j$
is smaller than $\frac{|span_j|}{\lambda}$) that make the problem difficult.
(we should point out that Chuzhoy et al. \cite{COR06} also considered this special case and presented a DP algorithm
with run time $O(n^{\rm{Poly}(\lambda)}T^4)$ however, their DP table is indexed by integer points on the time-line and
the polynomial dependence on $T$, which can be exponential in $n$, is unavoidable).
The idea of the dynamic program of the next theorem is the basis of the more general case that we will prove later that handles
``loose'' and ``tight'' jobs together but the following theorem is easier to understand and follow and we present it as a warm-up
for the main theorem.

\begin{theorem}\label{thm:relatively_big_processing_times}
If for all $j\in J$ in $\cI'$, $p_j \geq \frac{\lvert span_j \rvert}{\lambda}$ then there is a dynamic programming
 algorithm that finds a canonical solution for instance $\cI'$ with total profit $\opt'$ 
in time $O(\varepsilon^{-1}n^{\varepsilon^{-2}c}\log T)$.
\end{theorem}

\begin{proof}
Recall that $k=\log_q T$ and observe that for each $0 \leq i \leq k-1$ and each $j \in \cJ_i$:
 $\lambda \cdot \ell_i \leq \lvert span_j \rvert \leq \lambda p_j$, so $\ell_i \leq p_j$. Now if we somehow know 
$\OPT' \cap \cJ_0$ and $\OPT' \cap \cJ_1$ and remove the rest of jobs in $\cJ_0$ and $\cJ_1$, then the remaining jobs 
(which are all in $\cJ_{i\geq 2}$)
have intersection with exactly one interval in $I_0$ (recall we have no span-crossing or position-crossing jobs), 
hence we would have $q+1$ many independent sub-problems (defined on the $q+1$ sub-intervals partitioned in level $0$)
with jobs from $\cJ_{i\geq 2}$. 

So our first task is to ``guess'' the jobs in $\OPT'\cap (\cJ_0 \cup \cJ_1)$ (as well as their positions) and 
then remove the rest of the jobs
in $\cJ_0\cup \cJ_1$ from $J$ as well as the jobs whose span is crossing any of the intervals in $I_0$; 
then recursively solve the problem
on independent sub-problems obtained for each interval in $I_0$ together with the jobs whose spans are entirely within such interval.
In order to guess the positions of jobs in $\OPT'\cap (\cJ_0 \cup \cJ_1)$ we use the fact that each job can start at a slack time.
Since jobs in $\cJ_0\cup \cJ_1$ have size at least $\ell_1=T/q^2$, we can have at most $q^2$ of them in a solution. 
We guess a set $S$ of size at most $q^2$ of such jobs and a schedule for them; there are at 
most $|\Psi|^{q^2}=n^{O(q^2c)}$ choices for the schedule of $S$. Then we remove
the rest of $\cJ_0$ and $\cJ_1$ from $J$ for the rest of our dynamic programming. The guessed schedule of $S$ defines
a vector $\vec{v}$ of blocked spaces (those that are occupied by the jobs from $S$) and 
for each interval $a_{0,t}$, the projection of vector $\vec{v}$ in interval $a_{i,t}$, denote it by $\vec{v}_t$, has dimension at most $q$
($a_{0,t}$ has length $\ell_0 = T/q$ and each job in $S$ has length at least $\ell_1 = T/q^2$).
We pass each such vector $\vec{v}$ to the corresponding sub-problem.




Consider an interval $a_{i,t} \in I_i$ for some $0 \leq i \leq k$ and $0 \leq t \leq \frac{T}{\ell_i}$. Recall
that the set of jobs $j \in J$ whose span is completely inside $a_{i,t}$ is $J(a_{i,t})$. Because of the assumption 
of no span-crossing jobs, for each job $j \in J \setminus J(a_{i,t})$, if its span has
intersection with $a_{i,t}$, then it would be in $\cJ_{i'}$ for some $i' \leq i+1$
(jobs from $\cJ_{i+2}$ are entirely within one interval of level $I_i$) and $\lvert span_j \rvert$ would be at least
 $\lambda \ell_{i+1}$, and hence $p_j\geq \ell_{i+1}$. Thus we can have at most $\ell_i/\ell_{i+1}=q$ such jobs.
Assume we have a guessed vector $\vec{v}$ of length $q$ where each entry of the vector denotes 
the start time as well as the end time of one of such jobs. This vector describes the sections of $a_{i,t}$ that
are blocked for running such jobs from $J\setminus J(a_{i,t})$. The number of guesses for such vectors $\vec{v}$ is
at most $n^{2q(c+1)}$ based on the bounds on the number of slack times. Given $\vec{v}$ and $J(a_{i,t})$ we want to schedule
the jobs of $J(a_{i,t})$ in the free (unblocked by $\vec{v}$) sections of $a_{i,t}$.

Now we are ready to precisely define our dynamic programming table. For each $a_{i,t}$ and for each $q$-dimensional vector
 $\vec{v}$, we have an entry in our DP table $A$. This entry, denoted by $A[a_{i,t}, \vec{v}]$, 
will store the maximum throughput for an schedule of jobs running during interval $a_{i,t}$, 
using jobs in $J(a_{i,t})$ by considering the free slots defined by $\vec{v}$. 
The final solution would be $\max_{S}\{\sum_{t} A[a_{0,t}, \vec{v}_t]+|S|\}$, where the max is taken over
all guesses $S$ of jobs from $\cJ_0\cup\cJ_1$ and $\vec{v}_t$ is the blocked area
of $a_{i,t}$ based on $S$.

The base case is when $a_{i,t}$ has only constantly many release/deadline times. Given that we have also only constantly many
processing times and $\vec{v}$ defines at most $q$ many sections of blocked (used by bigger jobs) areas, then
using Theorem \ref{theo:base} we can find a $(1-O(\eps))$-approximation in time $\Gamma$, where
$\Gamma$ is the running time of the PTAS for Theorem \ref{theo:base}.

We can bound the size of the table as follows.
First note that we do not really need to continue partitioning an interval $a_{i,t}$ if there are at most $O(1)$ many 
distinct release times and deadlines within that interval, since this will be a base case of our 
dynamic program. So the hierarchical decomposition of intervals $I_0,I_1,\ldots,I_k$ will actually 
stop at such an interval $a_{i,t}$ when there are at most $O(1)$ release times and deadlines. 
Therefore, at each level $I_i$ of the random hierarchical decomposition, there
are at most $O(n)$ intervals in $I_i$ that will be decomposed into $q$ more intervals in $I_{i+1}$ (namely those that have
at least a constant number of release times and deadlines within them). 
Thus the number of intervals at each level $I_i$
is at most $O(nq)$ and the number of levels is at most $k=\log_q T$. Therefore, 
the total number of intervals in all partitions is bounded by $O(knq)$.
To bound the size of the table $A$, each $\vec{v}$ has $n^{2q(c+1)}$ many options, based on the fact that we have
at most $n^{c+1}$ many choices of start time and end time (from the set $\Psi$ of slacks) 
for each of the $q$ dimensions of $\vec{v}$. 
Also as argued above, there are $O(knq)$ many intervals $a_{i,t}$ overall. So the size of table is at most $kqn^{O(qc)}$.

Now we describe how to fill the entries of the table. 
To fill $A[a_{i,t}, \vec{v}]$ for each $0 \leq i \leq k-1$ and $0 \leq t \leq \frac{T}{\ell_i}$, suppose $a_{i,t}$ is divided 
into $q$ many equal size intervals $a_{i+1, t'+1}, ..., a_{i+1, t'+q}$ in $I_{i+1}$. We first guess a subset $\tilde{J_{i,t}}$ 
of jobs from $\cJ_{i+2} \cap J(a_{i,t})$, to be processed during interval $a_{i,t}$ consistent with free slots defined by $\vec{v}$. 
This defines a new vector $\vec{v}'$ that describes the areas blocked by jobs guessed recently as well as those blocked by $\vec{v}$.
Projection of $\vec{v}'$ onto the $q$ intervals $a_{i+1, t'+1}, ..., a_{i+1, t'+q}$ defines
$q$ new vectors $\vec{v}'_1, ..., \vec{v}'_q$. Now we check the sum of 

\[ A[a_{i+1, t'+1}, \vec{v}'_1]+A[a_{i+1, t'+2}, \vec{v}'_2] + ... + A[a_{i+1, t'+q}, \vec{v}'_q] + |\tilde{J_{i,t}}|\] 

We would choose the $\tilde{J_{i,t}}$ which maximizes the above sum. Observe that jobs in $J(a_{i,t})\setminus\cJ_{i+2}$
have length at most $\ell_{i+3}$ and because we have no position-crossing jobs, each of them is inside one of
intervals $a_{i+1, t'+1}, ..., a_{i+1, t'+q}$ and would be considered in sub-problems.

Note that to fill each entry $A[a_{i,t}, \vec{v}]$ the number of jobs from $\cJ_{i+2}$ possible to be processed in $a_{i,t}$
 would be at most $q^2$, because of their lengths. So the total number of guesses would be at most $n^{O(q^2c)}$. 
This means that we can fill the whole table in time at most $kqn^{O(q^2c)}$, where $q=1/\eps^2$ and $k=\log_q T$.
\end{proof}


Considering Theorem \ref{thm:relatively_big_processing_times}, we next show how to handle ``loose'' jobs,
i.e. those for which $p_j < \frac{|span_j|}{\lambda}$.
Recall that for each $0 \leq i \leq k$ and for each $j \in \cJ_i$, if $span_j$ has intersection with intervals 
$a_{i, t_j}, a_{i, t_j+1}, ..., a_{i, t'_j}$ of $I_i$, then we denote $span_j \cap a_{i, t_j}$ and $span_j \cap a_{i, t'_j}$
 as the head and tail of (span of) $j$, respectively. 
Our next (technical) lemma states that if we reduce the span of each loose job by removing its head and tail then there is still
a near optimum solution for $\cI'$.
More specifically, for each job loose $j\in \cJ_i$ ($p_j \leq \frac{\lvert span_j \rvert}{\lambda}$), 
whose span has intersection with intervals $a_{i, t_j}, a_{i, t_j+1}, ..., a_{i, t'_j}$ of $I_i$, we replace its
release time to start at the beginning of $a_{i,t_j+1}$ and its deadline to be end of $a_{i,t'_j-1}$;
so $span_j$ will be replaced with with $span_j \setminus (a_{i, t_j} \cup a_{i, t'_j})$. 
Let this new instance be called $\cI''$.
Note that a feasible solution for instance $\cI''$ would be still a valid solution for $\cI'$ as well.

\begin{lemma}\label{thm:head_tail_reduction}
Starting from $\cI'$, let $\cI''$ be the instance obtained from removing the head and tail part of $span_j$ for each 
job $j \in J$ with $p_j \leq \frac{\lvert span_j \rvert}{\lambda}$. Then 
there is a canonical solution for $\cI''$ with throughput at least $(1 - 120\varepsilon c)\opt'$.
\end{lemma}

\begin{proof} 
We will prove the following important key lemma in Section \ref{sec:details}. 

\begin{lemma}[\bf Head and tail cutting] \label{lmm:head_tail_reduction}
Consider any fixed processing time $p \in P$. Start with instance $\cI'$ and remove only the head
(or only the tail) part of $span_j$ for all jobs $j \in J$ with $p_j = p \leq \frac{\lvert span_j \rvert}{\lambda}$. 
Then there is a solution for the 
remaining instance with profit at least $(1-\frac{60}{\lambda})\opt'$.
\end{lemma}

Considering Lemma \ref{lmm:head_tail_reduction}, the proof of
Lemma \ref{thm:head_tail_reduction} would be easy. We just need to apply Lemma \ref{lmm:head_tail_reduction} 
for all $c$ many distinct processing
 times $p \in P$ and for both "head" and "tail". Then the total loss for removing all head and tail parts would be 
$\frac{60}{\lambda} \cdot 2c = 120\varepsilon$ fraction:
\begin{equation*}
\opt(\cI'') \geq (1-\frac{60\times 2c}{\lambda}) \opt' \geq (1-120\varepsilon c)\opt'.
\end{equation*}

\end{proof}

The next theorem together with Lemmas \ref{lmm:span_crossing}, \ref{lmm:span_optimum_crossing}, and
\ref{lmm:head_tail_reduction} will help us to complete the proof.

\begin{theorem}\label{theo:main2}
There is a dynamic programming algorithm that finds an optimum solution for instance $\cI''$ 
in time $\varepsilon^{-3}n^{O(\varepsilon^{-6}c)}\log T$.
\end{theorem}

Before presenting the proof of this theorem we show how this can be used to prove Theorem \ref{theo:main} for $m=1$.

\begin{proof} [Proof of Theorem \ref{theo:main}]
Starting from instance $\cI$ we first reduced it to instance $\cI'$
at a loss of $1-O(\varepsilon)$. Then remove the head and tail part of the span for all the
loose jobs to obtain instance $\cI''$.
Based on Lemma \ref{lmm:head_tail_reduction}, we only loose a factor of $(1-O(\varepsilon c))$ compared to optimum of $\cI'$.
Theorem \ref{theo:main2} shows we can actually find an optimum canonical solution to instance $\cI''$. This solution
will have value at least $(1-O(\varepsilon c))\opt$ using Lemmas \ref{lmm:span_crossing}, \ref{lmm:span_optimum_crossing}, and
\ref{lmm:head_tail_reduction}. To get a $(1-\varepsilon')$-approximation we set $\varepsilon'=\varepsilon/c$ in 
Theorem \ref{theo:main2}.
The run time will be $c^3\eps^{-3}n^{O(\varepsilon'^{-6}c^7)}\log T$.

\end{proof}

Now we prove Theorem \ref{theo:main2}.

\begin{proof}
The idea of the proof is similar to that of Theorem \ref{thm:relatively_big_processing_times}. However, the presence of
``loose'' jobs needs to be handled too. Suppose $j\in\cJ_i$ is a loose job, so $\lambda\ell_i\leq |span_j| < \lambda\ell_{i-1}$ 
and $p_j\leq \frac{span_j}{\lambda}< \ell_{i-1}$. We break these loose jobs into two categories.
For the loose jobs that $p_j < \ell_{i+1}$, because they are not
position-crossing, their position in the final solution will have intersection with at most one interval of $I_{i}$
(and so we can pass them down to lower sub-problems). But for loose jobs where $\ell_{i+1}\leq p_j <\ell_{i-1}$ we need to guess
them (similar to the tight jobs) and we can do the guessing since their size (relative to $\ell_i$) is big.
In order to handle these guesses, we add one more vector to the DP table, and we do the guess for two consecutive levels 
of our decomposition as we go down the DP. 

Suppose $P = \{p_1, p_2, ..., p_c\}$. For each interval $a_{i,t}$ ($0 \leq i \leq k$, $0 \leq t \leq \frac{T}{\ell_i}$), 
$q^2$-dimensional vector $\vec{v}$ (where $0\leq v_i\leq n$), $(qc)$-dimensional vector $\vec{u} = (u_{1,1},\ldots, u_{q,c})$,
where each $u_{\gamma,\sigma}$, $0\leq u_{\gamma,\sigma}\leq n$,
we have an entry in our DP table $A$. Suppose $a_{i,t}$ is partitioned 
into intervals $a_{i+1, t'+1}, ..., a_{i+1, t'+q}$ in $I_{i+1}$.
Entry $A[a_{i,t}, \vec{v}, \vec{u}]$, will store the maximum throughput of a schedule in interval $a_{i,t}$ by
selecting subsets of jobs from the following two collections of jobs:
\begin{itemize}
\item $J(a_{i,t})\cap \cJ_{\geq(i+2)}$
\item $u_{\gamma,\sigma}$ many jobs with processing time $p_{\sigma}$ where $p_{\sigma}<\ell_{i+2}$
whose span is the entire interval $a_{i+1,t'+\gamma}$, for each $1\leq \gamma\leq q$, and $1\leq \sigma\leq c$.
\end{itemize}
by considering the free slots defined by vector $\vec{v}$ (that describes blocked spaces by jobs of higher levels).

Vector $\vec{u}$ is defining the sets of jobs from loose jobs (from higher levels of DP table) whose span was initially much larger
than $\ell_{i+1}$, the guesses we made requires them to be scheduled 
in interval $a_{i+1,t'+\gamma}$ (of length $\ell_{i+1}$) and hence their span is the entire interval $a_{i+1,t'+\gamma}$.
Like before, $\vec{v}$ is defining the portions of the interval which are already used by
 bigger jobs (that are guessed at the higher levels), and for similar reasons as in 
Theorem \ref{thm:relatively_big_processing_times}, we
 only need to consider $\vec{v}$'s of size at most $q^2$ and each job listed in $\vec{v}$ 
will be denoted by its start position and end position (so there is $O(|\Psi|^{2q^2})=n^{O(q^2c)}$ possible values
for $\vec{v}$).

Similar to Theorem \ref{thm:relatively_big_processing_times}, suppose we start at $I_0$.
We guess a subset of tight jobs from $\cJ_0$ to decide on their
schedule. Note that tight jobs will have $p_j\geq \ell_0$.
We also need to guess (and decide on their schedule) those ``loose'' jobs $j\in\cJ_0$ where $p_j\geq \ell_2=T/q^3$
(since their position may cross more than one $I_1$ intervals in the final solution). So we guess a set $S_0\subseteq \cJ_0$
 with $|S_0|\leq q^3$ of jobs $j$ where $p_j\geq\ell_2$ and a feasible schedule for them. 
This will take care of guessing tight and those loose jobs of $\cJ_0$ with $p_j\geq\ell_2$. We need to do similarly for
jobs from $\cJ_1$, i.e. we need to guess a set of tight jobs $j$ 
from $\cJ_1$ (note that for them $p_j\geq\ell_1$) 
and also guess (and decide on their schedule) those ``loose'' jobs $j\in\cJ_1$ with $p_j\geq \ell_2$.
To do so, we guess a set $S_1\subseteq\cJ_1$ of jobs $j$ where $p_j\geq\ell_2=T/q^3$ and a 
feasible schedule for them (given the guesses for $S_0$); 
note that $|S_0\cup S_1|\leq q^3$ (since all of $S_0\cup S_1$ must fit in $[0,T]$).
For each such guess, their schedule projects a vector of blocked
spaces (occupied time of machine). This will be vector $\vec{v}$. The projection of $\vec{v}$ to each interval $a_{0,t}$
will be $\vec{v}_t$ which is the blocked area of $a_{0,t}$. Note that although $\vec{v}$ has up to $q^3$ blocks, 
each $a_{0,t}$ can have at most $q^2$ blocks since each block has size at least $\ell_2=T/q^3$ and each $a_{0,t}$ has
size $\ell_0=T/q$.

For all the other jobs in $\cJ_0\cup \cJ_1$ that have $p_j<\ell_2$, because they are not position-crossing, we can assume
their position (in the final solution) has intersection with only one interval of $I_1$.
For all these jobs of $\cJ_0\cup\cJ_1$, we use the assumption that there is a near
optimum solution in which they are not scheduled in their head or tail. 
So for the jobs in $\cJ_0\cup\cJ_1$ with processing time less than $\ell_2$ we can re-define their 
span to a guessed interval of $I_1$; these guesses
define the $qc$-dimensional vectors $\vec{u}_t$ for each of the $q$ sub-intervals of $a_{0,t}$ at level
$I_1$ (how many loose jobs from $\cJ_0\cup\cJ_1$ with $p_j<\ell_2$
have their span redefined to be one of sub-intervals of $a_{0,t}$).
The final solution will be $\max_{S_0,S_1}\{\sum_t A[a_{0,t},\vec{v}_t,\vec{u}_t]+|S_0\cup S_1|\}$, where the max
is taken over all guesses $S_0\subseteq\cJ_0$, $S_1\subseteq\cJ_1$ and $\vec{u}_t$ as described above.

To bound the size of the table, as argued before, we would have at most $O(knq)$ many intervals in all of $I_0,I_1,\ldots,I_k$. 
For each of them we consider a table entry for at most $n^{O(q^2c)}$ many vectors $\vec{v}$, $n^{O(qc)}$ many vectors $\vec{u}$.
So the total size of the table would be $(kq)n^{O(q^2c)}$.

Like before, the base case is when interval $a_{i,t}$ has $O(1)$ many release times and deadlines. These base cases
$A[a_{i,t},\vec{v},\vec{u}]$ can be solved using Theorem \ref{theo:base} for each vector $\vec{v}$ and $\vec{u}$.

To fill $A[a_{i,t}, \vec{v}, \vec{u}]$ in general (when $0\leq i\leq k$ and $0 \leq t \leq \frac{T}{\ell_i}$ and there
are more than $O(1)$ many release times and deadlines in $a_{i,t}$), suppose $a_{i,t}$ is 
divided into $q$ many equal size intervals $a_{i+1, t'+1}, ..., a_{i+1, t'+q}$ in $I_{i+1}$. 
What we decide at this level is:

\begin{itemize}
\item make a decision for all the jobs $j\in \cJ_{i+2}\cap J(a_{i,t})$; those that are bigger
than $\ell_{i+3}$ will be scheduled or dropped by making a guess; the rest we narrow down their span (guess) 
to be one of the lower level sub-intervals
of $a_{i,t}$ and will be passed down as $\vec{u}'$ to sub-problems below $a_{i,t}$;

\item make a decision for jobs in $\vec{u}$: those that are bigger than $\ell_{i+3}$ will be scheduled or dropped;
the rest we narrow down their span (by a guess) to be one of the lower level sub-intervals of $a_{i,t}$
\end{itemize}

As in the case of $I_0$, we need to guess a set of tight jobs from $\cJ_{i+2}\cap J(a_{i,t})$ and some loose jobs
$j$ with $p_j\geq\ell_{i+3}$ and their positions to be processed in $a_{i,t}$ (considering the blocked areas defined by $\vec{v}$).
Let $S_0$ with $s_0=|S_0|$ be this guessed set. Note that $s_0\leq q^3$ since $p_j \geq \ell_{i+3}=\ell_i/q^3$.
Also for each non-zero $u_{\gamma,\sigma}$ where $p_{\sigma} \geq \ell_{i+3}$ we guess how many of 
those $u_{\gamma,\sigma}$ many jobs should be scheduled and where exactly in $a_{i+1,t'+\gamma}$ (consistent with $\vec{v}$ and $S_0$); 
let $S_1$ be this guessed subset and $|S_1|=s_1$. Note that $s_0+s_1\leq q^3$ and 
there are at most $|\Psi|^{2q^3}$ possible guesses for $S_0$ and $S_1$ together with their positions; thus a total of $n^{O(q^3c)}$ 
possible ways to guess $S_0\cup S_1$ and guess their locations in the schedule.
Then for each possible pair of such guessed sets $S_0,S_1$ we compute the resulting $\vec{v}'$; this defines
the space available for the rest of the jobs in $J(a_{i,t})\cap\cJ_{\geq i+3}$, and those 
defined by $\vec{u}$ where $p_j<\ell_{i+3}$ after blocking the space defined by $\vec{v}$ and the
space occupied by the pair of guessed sets $S_0,S_1$ above. We divide $\vec{v}'$ into 
$q$ many vectors $\vec{v}'_1, ..., \vec{v}'_q$, (as we divided $a_{i,t}$ into $q$ intervals).

We also change $\vec{u}$ to $\vec{u}'$ by setting all the entries of $u_{\gamma,\sigma}$ with $p_{\sigma} \geq \ell_{i+3}$ to zero 
and guess how to distribute $\vec{u}'$ into $q$ many $(qc)$-dimensional vectors $\vec{u}'_1, ..., \vec{u}'_q$ such that 
$\vec{u}'_1 + \vec{u}'_2 + ... + \vec{u}'_q = \vec{u}'$, where $\vec{u}'_\gamma$ is 
describing the number of jobs of different sizes whose span is re-defined to be one of the sub-intervals of $a_{i+1,t'+\gamma}$
at level $I_{i+2}$. The number of ways to break $\vec{u}'$ into $\vec{u}'_1,\ldots,\vec{u}'_q$ is bounded by $n^{O(q^2c)}$.

For all the other jobs in $\cJ_{i+2}\cap J(a_{i,t})$ that have $p_j<\ell_{i+3}$, because they are not position-crossing, we can
 assume
their position (in the final solution) has intersection with only one interval of $I_{i+2}$.
We also use the assumption that there is a near optimum solution in which they are not scheduled in their head or tail. 
So for the jobs in $\cJ_{i+2}\cap J(a_{i,t})$ with processing time less than $\ell_{i+3}$ we can re-define their 
span to a guessed sub-interval of $a_{i+1,t'+\gamma}$ at level $I_{i+2}$; these guesses
define the $qc$-dimensional vectors $\vec{w}_{\gamma}$ for each interval $a_{i+1,t' + \gamma}$ (how many loose jobs from 
$\cJ_{i+2}\cap J(a_{i,t})$ with $p_j<\ell_{i+3}$
have their span redefined to be one of the $q$ sub-intervals of $a_{i+1,t'+\gamma}$ at level $i+2$). 
Observe that, by only knowing how many of $w_{\sigma}$ many jobs with processing times
$p_{\sigma}$ are scheduled in each interval $a_{i+1, t'+1}, ..., a_{i+1, t'+q}$ in the optimum solution, we would be able to detect
which job is in which interval. The reason is that we know for each $p_{\sigma} \in P$, all jobs with processing time $p_{\sigma}$
are scheduled based on earliest deadline first rule, which basically says that at any time when there are two jobs with the same 
processing time available the one with earliest deadline would be scheduled first.

Note that the jobs in $J(a_{i,t})\cap \cJ_{\geq(i+3)}$ all have processing time at most $\ell_{i+3}$ and 
their spans are completely inside one of intervals $a_{i+1, t'+1}, ..., a_{i+1, t'+q}$. 
These jobs will be passed down to the corresponding smaller sub-problems.
So for each given $\vec{v}$ and $\vec{u}$, we consider all guesses $S_0,S_1$ and consider the resulting $\vec{u}',\vec{v}'$
and any possible way of breaking $\vec{u}',$ and $\vec{w}$ into $q$ parts, we check:

\[ A[a_{i+1, t'+1}, \vec{v}'_1, \vec{u}'_1 + \vec{w}_1]+A[a_{i+1, t'+2}, \vec{v}'_2, \vec{u}'_2 + \vec{w}_2] + \ldots + 
A[a_{i+1, t'+q}, \vec{v}'_q, \vec{u}'_q + \vec{w}_q] + s_0 + s_1,\] 

where $s_0,s_1$ are the sizes of the subsets $S_0,S_1$ of jobs with processing time $p_j \geq \ell_{i+3}$ 
guessed from $J(a_{i,t}) \cap \cJ_{i+2}$ and those from $\vec{u}$ with processing time $p_j \geq \ell_{i+3}$. 
We would choose the maximum over all guesses $S_0\subseteq \cJ_{i+2}\cap J(a_{i,t})$, $S_1$, 
and all possible ways to distribute jobs with $p_j < \ell_{i+3}$ to create $\vec{u}'_{\gamma}$ and $\vec{w}'_{\gamma}$ as 
described above.

Note that to fill each entry $A[a_{i,t}, \vec{v}, \vec{u}]$ the number of jobs from $\cJ_{i+2} \cap J(a_{i,t})$ plus jobs from
 $\vec{u}$ with processing time bigger than $\ell_{i+3}$ possible to be processed in $a_{i,t}$ would be at most $q^3$, because
 of their lengths. So we could have at most $n^{O(q^3c)}$ many different $\vec{v}'$ to consider. For $\vec{u}'$ and
 $\vec{w}$ we would have at most $n^{O(q^2c)}$ many ways to distribute each of them into $q$ many $qc$-dimensional vectors.
This means that we can fill the whole table in time at most $\Gamma k q n^{O(q^3c)}=n^{O(\varepsilon^{-6}c)}\log_q T$,
where $\Gamma$ is the running time of the PTAS for Theorem \ref{theo:base}, which is at most 
$2^{\eps^{-1}\log^{-4}(1/\eps)}+{\rm Poly}(n)$. So the total time will be $n^{O(\varepsilon^{-6}c)}\log T$.
\end{proof}

\subsection{Extension to $m=O(1)$ Machines}
We show how to extend the result of Theorem \ref{theo:main} to $m=O(1)$ machines.
We first do the randomized hierarchical decomposition of time line $[0,T]$ and define the classes of jobs $\cJ_0,\cJ_1,\ldots$
as before. Lemmas \ref{lmm:span_crossing} and \ref{lmm:span_optimum_crossing} can be adjusted to show that there is a 
solution with no span-crossing or position-crossing jobs of value at least $(1-O(\varepsilon))\opt$.
Lemma \ref{thm:head_tail_reduction} still holds for each machine. So we only need to explain how to change the DP
for Theorem \ref{theo:main2}. Our dynamic program will be similar, except that 
for each interval $a_{i,t}$ sub-problems are defined based on 
$m$ vectors $\vec{v}^1,\vec{v}^2,\ldots,\vec{v}^m$ corresponding to the blocked areas of the interval over machines $1,\ldots,m$
as well as vector $\vec{u}$.
The sub-problems are stored in entries $A[a_{i,t},\vec{v}^1,\vec{v}^2,\ldots,\vec{v}^m,\vec{u}]$ where each $\vec{v}^{i'}$
is a $q^2$-dimensional vector describing the blocked areas of $a_{i,t}$ on machine $i'$ using jobs from $\cJ_{\leq i+2}$.
Vector $\vec{u}$ as before is a $(qc)$-dimensional vector describing (for each $1\leq \sigma \leq c$) the number of jobs of
size $p_\sigma$ that their span is redefined to one of the $q$ sub-intervals that $a_{i,t}$ will be divided into, on any of the machines.
So the number of sub-problems will be $(kn)n^{(m+c)q^2}$.
At each step of the recursion, to fill in the entry $A[a_{i,t},\vec{v}^1,\vec{v}^2,\ldots,\vec{v}^m,\vec{u}]$ we have to make
similar guesses as before, except that now we have to decide on which of the $m$ machines we schedule them.
For the sets $S_0,S_1$ guessed from tight jobs and loose jobs from $\cJ_{i+2}\cap J(a_{i,t})$, we have $|\Psi|^{2q^3}$ guesses
and for each of guesses another $m$ options to decide the machines. So we will have $n^{O(mq^3c)}$ guesses.
The number of guesses to break $\vec{u}$ to $\vec{u}'_1,\ldots,\vec{u}'_q$ will be the same. The rest of the computation of 
the entry is independent of the machines as we don't schedule any more jobs at this point.
Hence, the total complexity of computing the entries of the DP table will be 
$O(\Gamma\varepsilon^{-2}kn^{O(mcq^3)})=\varepsilon^{-3}n^{O(mc\varepsilon^{-6})}\log T$ (again noting
that $\Gamma$ being the running time of algorithm of Theorem \ref{theo:base}) and we obtain a 
$(1-O(c\varepsilon))$-approximation.
For fixed $m$ and $c$ and for a given $\varepsilon'>0$ one can choose $\varepsilon=\varepsilon'/c$ to obtain a 
$(1-\varepsilon')$-approximation in time $n^{O(mc^7\varepsilon'^{-6})}\log T$.

If all $p_j$'s are bounded polynomially in $n$ then we can also use Theorem \ref{theo:main} to
obtain a bicriteria $(1-\epsilon,1+\epsilon)$ quasi-polynomial time approximation. 
For simplicity consider the case of a single machine ($m=1$). Given $\varepsilon'>0$, we scale the processing times
up to the nearest power of $(1+\varepsilon')$. So we will have $c=O(\log n/\varepsilon')$ many distinct processing times.
We the run the algorithm of Theorem \ref{theo:main} with $\varepsilon=\frac{\varepsilon'}{c}=\frac{\varepsilon'^2}{\log n}$.
This will give a $(1+O(\varepsilon c))$-approximation which we can run on a machine with $(1+\varepsilon')$-speedup to compensate for
the scaled-up processing times (so each scaled job will still finish by its deadline on the faster machine).
Since $\varepsilon c = \frac{\varepsilon'^2}{\log n}\cdot \frac{\log n}{\varepsilon'}=\varepsilon'$, 
we obtain a $(1-\varepsilon')$-approximation
on $(1+\varepsilon')$-speedup machine in time $n^{O(\varepsilon'^{-13}\log^7 n)}$ (as mentioned earlier a stronger form of
this, i.e. for weighted setting was already known \cite{ILMT15}).

\section{Cutting heads and tails: Proof of Lemma \ref{lmm:head_tail_reduction}}\label{sec:details}
We focus on optimum solution $O=\OPT'$ and show how to modify $O$ so that none of the jobs in the modified instance
are scheduled in their head part without much loss in the throughput. For simplicity, we assume that 
$J$ only contains the set of jobs scheduled in $O$.
We basically want to construct another solution $O''$ by changing $O$ such that in $O''$ the position of each loose job
with processing time $p$ has no intersection with its "head" part and at the same time its total profit is 
still comparable to $O$, which allows us to remove "head" part and still have a feasible solution with the 
desired total profit.

For each job $j$, recall that $span_j=[r_j,d_j]$, and if $j\in \cJ_i$ and $span_j$ has intersection with 
$a_{i, t_j}, ..., a_{i, t'_j}$ from $I_i$ then $head_j=span_j\cap a_{i,t_j}$ and $tail_j=span_j\cap a_{i,t'_j}$.
We let $\overline{span}_j=span_j-(head_j\cup tail_j)$ be the reduced span of $j$. 
Our goal is to modify $O$ so that every loose job $j$ is scheduled in $O$ in $\overline{span}_j$.
The idea of the proof is to move each loose job $j$ with processing time $p$ scheduled in its head (or tail) to be re-scheduled in $\overline{span}_j$
if there is empty space for it there. If not, and if we can remove some larger (w.r.t. processing time) jobs in $\overline{span}_j$
to make room for $j$ and possibly other loose jobs whose head is in $a_{i,t_j}$ we do so.
Otherwise, it means that the entire $\lambda$ intervals starting from $a_{i,t_j}$ which $span_j$ has intersection with is
relatively packed with jobs of size $p$ or smaller. We want to argue that in this case even if we remove $j$ (and all other
loose jobs in $a_{i,t_j}$) we can ``charge'' them to the collection of many jobs scheduled in the next $\lambda$ intervals;
hence the loss will be relatively small. 
However, we cannot do this simple charging argument since the intervals to which we charge (for the jobs removed) are not all
disjoint; hence a job that remains might be charged multiple times (due to the hierarchy of the intervals we have defined).
Nevertheless, we show a careful charging scheme that will ensure the total loss for jobs, that cannot be rescheduled in their
reduced span, is still relatively small.

\begin{proof}
Consider $O=\OPT'$ and assume that $J$ is simply the set of jobs in $O$.
We focus on the loose jobs of size $p$ that their position in $O$ has intersection with their ``head''
(argument is similar for the case of ``tail'' we just do the reverse order).
We traverse all the loose jobs of size $p$ in $J$ in the order of their position in $O$ from the latest to the earliest. 
For each such job $j \in J$ assume $j \in \cJ_i$ for some $0 \leq i \leq k$ and $span_j$ has intersection with 
$a_{i, t_j}, ..., a_{i, t'_j}$ from $I_i$. Note that since $j\in \cJ_i$ it means $t'_j-t_j\geq \lambda$.
While traversing $j$ if its position in $O$ has intersection with $head_j$ we add it to set $X_{i, t_j}$
(which is initially empty) corresponding to interval $a_{i, t_j} \in I_i$ and try to move it to $\overline{span}_j$ if possible
 (without changing the position of any other job). This means if there is empty space in $\overline{span_j}$ we try
to re-schedule $j$ there. If this is not possible, then
 temporarily remove it from $O$ (to make room for the rest of the jobs currently running in their head) and add it to set 
$X'_{i, t_j}$ (which is initially empty too).

After changing the position of some loose jobs and removing some others, it is obvious that the position of each 
scheduled loose job of size $p$ has no intersection with its head in the current solution which we denote by $O'$. 
Observe that for each interval 
$a_{i,t} \in I_i$ for $0 \leq i \leq k$ and $0 \leq t \leq \frac{T}{\ell_i}$, we have $X'_{i, t} \subseteq X_{i,t}$ and
 $\lvert X'_{i,t} \rvert = x'_{i,t} \leq \lvert X_{i,t} \rvert = x_{i,t}$. Also if $x'_{i,t} > 0$, then there is no empty 
space for a job with processing time $p$ in the following $\lambda - 1$ intervals of $I_i$, i.e. if we define
 $Y_{i,t} = a_{i, t+1} \cup ... \cup a_{i, t+\lambda-1}$, there is no empty space of size $p$ in $Y_{i,t}$. 
This uses the fact that for any job like $j$ whose head is $a_{i,t}$, its span contains all of $Y_{i,t}$. So $x'_{i,0}>0$ means
there are such jobs of size $p$ (whose head is in $a_{i,t}$) and they could not be moved to any space in $Y_{i,t}$.

Consider interval $a_{i,t}$ for any $0 \leq i \leq k$ and $0 \leq t \leq \frac{T}{\ell_i}$. We
 define $y_{i,t} = \frac{\lvert Y_{i,t} \rvert}{p} = \frac{(\lambda - 1) \cdot \ell_i}{p}$, and $A_{i,t}$ as the set consisting
of all $a_{i',t'}$ such that $Y_{i,t} \cap Y_{i', t'} \neq \emptyset$ and

\begin{itemize}
\item $i'>i$, or
\item $i' = i$ and $t' > t$. 
\end{itemize}

So those in $A_{i,t}$ are the intervals $a_{i',t'}$ whose $Y$ set has overlap with that of $a_{i,t}$ and either $a_{i',t'}$
is at a finer level of hierarchy, or is at the same level $i$ but at a later time.
We then partition $A_{i,t}$ into two sets $A_{i,t}^1$ and $A_{i,t}^2$:
\begin{itemize}
\item if $a_{i',t'} \subseteq a_{i, t}$ then $a_{i',t'} \in A_{i,t}^{1}$, 
\item else $a_{i',t'} \in A_{i,t}^2$.
\end{itemize}

Observe that for each $a_{i', t'} \in A_{i,t}^2$ we have $a_{i', t'} \subseteq Y_{i,t}$ and this means that removing any job from
 $a_{i',t'} \in A_{i,t}^2$ would make an empty room for a job in $X'_{i,t}$.

Next lemma would provide an important fact about intervals whose $Y$ parts are not disjoint and basically provides an upper
 bound on the number of jobs removed temporarily from all intervals in $A_{i,t}$ during the first phase while converting $O$
 to $O'$:

\begin{lemma}\label{lmm:XtoY}
For each $0\leq i \leq k$ and $0 \leq t \leq \frac{T}{\ell_i}$ with $x'_{i,t}>0$:
\begin{itemize}
\item $x'_{i,t} + \sum_{a_{i',t'} \in A_{i,t}^1} x'_{i', t'} \leq \frac{3}{\lambda} \cdot y_{i,t}$, 
\item $x'_{i,t} + \sum_{a_{i',t'} \in A_{i,t}^2} x'_{i', t'} \leq \frac{3}{\lambda} \cdot y_{i,t}.$
\end{itemize}
\end{lemma} 
We defer the proof of this lemma to later.
\begin{corollary}\label{cor:XtoY}
For each $0\leq i \leq k$, and $0 \leq t \leq \frac{T}{\ell_i}$ with $x'_{i,t}>0$:
\label{cor:XtoY}
\begin{equation*} 
x'_{i,t} + \sum_{a_{i',t'} \in A_{i,t}} x'_{i', t'} \leq \frac{6}{\lambda} \cdot y_{i,t} 
\end{equation*} 
\end{corollary} 

Next we traverse all intervals on a specific order and change $O'$ to $O''$ so that we can compare its total profit with 
$\OPT'$ while still no scheduled job has intersection with its "head" part. For each $i$ from $0$ to $k$ and for
 each $t$ from $0$ to $\frac{T}{\ell_i}$, if $x'_{i,t} > 0$ do the following:

If (and while) the processing time of the biggest job which is currently scheduled in $Y_{i,t}$ is more than $p$, 
 and $X'_{i,t}$ is not empty yet, remove that
 biggest job from $O'$, add it to set $R_{i,t}$ (which is initially empty) and add as many jobs from $X'_{i,t}$ to $O'$ as possible
 in the empty space which is just freed up by removing that big job. We repeat this as long as $X'_{i,t}\not=\emptyset$
and the size of the biggest job currently scheduled in $Y_{i,t}$ is larger than $p$.
Note that jobs in $X'_{i,t}$ all have processing time $p$ 
and able to be scheduled in whole $Y_{i,t}$ since their span contains $Y_{i,t}$. 
At the end, if $X'_{i,t}\not=\emptyset$ and the processing time of the biggest remaining job in $Y_{i,t}$ is 
no more than $p$ (or in the case it  was initially at most $p$), add all the remaining jobs in $X'_{i,t}$ to $R_{i,t}$, and
define $p'_{i,t}$ as the processing time of the smallest job in $R_{i,t}$ and set $\alpha_{i,t} = \lfloor\frac{p'_{i,t}}{p}\rfloor$.
Note that all the jobs remaining in $Y_{i,t}$ would have processing time at most $p'_{i,t}$.

Now we have our solution $O''$ which we claim has near optimum total profit. 
First observe that no loose job of size $p$ in $O''$ is scheduled having intersection with its head. Also, no job is moved to its head.
Note that for all $0 \leq i \leq k$ and 
$0 \leq t \leq \frac{T}{\ell_i}$, $R_{i,t}$ would contain all the jobs which are actually removed from optimum solution $O$:
\begin{equation*}
O = O'' \cup \bigcup_{i,t} R_{i,t}
\end{equation*}

Let's denote by $S_{i,t}$ the set of jobs scheduled inside $Y_{i,t}$ in solution $O''$ for each $0 \leq i \leq k$ and 
$0 \leq t \leq \frac{T}{\ell_i}$. Then the union of all these sets for all intervals would be a subset of $O''$:
\begin{equation*}
\bigcup_{i,t} S_{i,t} \subseteq O'' \quad \Rightarrow \quad \lvert \bigcup_{i,t} S_{i,t} \rvert \leq \lvert O'' \rvert
\end{equation*}

Our goal is to show that $\lvert \bigcup_{i,t} R_{i,t} \rvert \leq \frac{60}{\lambda} \lvert \bigcup_{i,t} S_{i,t} \rvert$ 
which completes the proof of Lemma \ref{lmm:head_tail_reduction}:
\begin{equation*}
\lvert O'' \rvert = \lvert O \rvert - \lvert \bigcup_{i,t} R_{i,t} \rvert \geq \lvert O \rvert - 
\frac{60}{\lambda} \cdot \lvert \bigcup_{i,t} S_{i,t} \rvert \geq \lvert O \rvert - \frac{60}{\lambda} \cdot \lvert O'' 
\rvert \geq (1-\frac{60}{\lambda}) \lvert O \rvert
\end{equation*}

The next lemma which upper bounds 
$\lvert R_{i,t} \rvert$ by a small fraction of $\lvert S_{i,t} \rvert$ can be proved using the ``simple'' charging scheme
explained at the beginning of this section. We defer the proof of this lemma to later.

\begin{lemma} \label{lmm:RtoS}
For each $0\leq i \leq k$ and $0 \leq t \leq \frac{T}{\ell_i}$ with $x'_{i,t}>0$:
\end{lemma}
\begin{equation*}
\lvert R_{i,t} \rvert \leq \frac{30}{\lambda} \lvert S_{i,t} \rvert 
\end{equation*}

This means that the number of jobs removed from $O$ for each interval $a_{i,t}$ (namely $\lvert R_{i,t} \rvert$), is at most 
$\frac{30}{\lambda}$ of the number of jobs scheduled in interval $Y_{i,t}$ (namely $\lvert S_{i,t} \rvert$). 
If it was the case that for any 
 two intervals $a_{i_1,t_1}$ and $a_{i_2, t_2}$, we have $S_{i_1, t_1} \cap S_{i_2, t_2} = \emptyset$,
 then Lemma \ref{lmm:RtoS} would be enough to complete the proof of Lemma \ref{lmm:head_tail_reduction}.
But the problem is that for any two different intervals $a_{i_1,t_1}$ and $a_{i_2,t_2}$, by definition, $R_{i_1, t_1}$ and
 $R_{i_2, t_2}$ are disjoint but $S_{i_1, t_1}$ and $S_{i_2, t_2}$ could have intersection. In other words we might have
 some intervals $a_{i_1, t_1}, a_{i_2, t_2}$ with $Y_{i_1, t_1} \cap Y_{i_2, t_2} \neq \emptyset$ which means 
$S_{i_1, t_1} \cap S_{i_2, t_2} \neq \emptyset$. The next lemma will help us to ``uncross'' those $Y$'s:

\begin{lemma}\label{lmm:Combine_Subset_All}
For each interval $a_{i, t}$ with $x'_{i,t}>0$, we can partition $A_{i,t}$ into two parts $A_1$ and $A_2$ such that
\begin{equation*}
\lvert R_{i,t} \cup \bigcup_{a_{i',t'} \in A_1} R_{i', t'} \rvert \leq \frac{60}{\lambda} \lvert S_{i,t} 
\setminus \bigcup_{a_{i', t'} \in A_2} S_{i', t'} \rvert
\end{equation*}
\end{lemma}

Using Lemma \ref{lmm:Combine_Subset_All} we can partition all intervals into a number of disjoint groups such that for each
group the number of total jobs removed from $O$ is a $\frac{60}{\lambda}$ fraction of the number of jobs scheduled 
in $O''$ in that group. 

Suppose $a_{i,t}$ is an interval with the lowest $i$ value (breaking the ties with equal $i$ by taking the smallest $t$) with $x'_{i,t}>0$.
Using Lemma \ref{lmm:Combine_Subset_All} we find some $A_1 \subseteq A_{i,t}$ and the first group $G_1$ 
 of intervals we define will be $G_1=\{a_{i,t}\} \cup A_1$. 
If we denote $R(G_1)= R_{i,t} \cup \bigcup_{a_{i',t'} \in A_1} R_{i', t'}$ and 
$S(G_1)= S_{i,t} \setminus \bigcup_{a_{i', t'} \in A_2} S_{i', t'}$ then using Lemma \ref{lmm:Combine_Subset_All}:
$|R(G_1)|\leq \frac{60}{\lambda}|S(G_1)|$.
Also $S(G_1)\cap S_{i',t'} = \emptyset$ for any
 $a_{i', t'} \notin A_1 \cup \{a_{i,t}\}$ for the following reason: if $a_{i',t'}\in A_2$ then clearly 
$S(G_1)\cap S_{i',t'}=\emptyset$ from definition of $S(G_1)$; if $a_{i',t'}\not \in A_1\cup A_2$ then $Y_{i',t'}$ has no intersection
with $Y_{i,t}$ and hence $S(G_1)\cap S_{i',t'}=\emptyset$.
Note that if $A_{i,t} = \emptyset$, then we can use Lemma \ref{lmm:RtoS}, we have $G_1 = \{a_{i,t}\}$ and 
$|R(G_1)|\leq \frac{60}{\lambda}|S(G_1)|$, holds for this case too.

So we can remove group $G_1$ along with the corresponding sets $R(G_1)$ and $S(G_1)$ and continue doing the same for the remaining
 intervals to construct the next group. Observe that at each step by removing a group of intervals, the 
remaining intervals are not changed and 
this allows us to be able to do the same process for them.
Finally we obtain a collection of groups $G_1,G_2,\ldots$ where for each $G_i$: $|R(G_i)|\leq \frac{60}{\lambda}|S(G_i)|$ 
and the sets $S(G_i)$'s are disjoint. Since $\bigcup_i S(G_i)$ is a subset of all jobs scheduled in $O''$ and $\bigcup_i R(G_i)$
is the set of all jobs removed from $O$ to obtain $O''$, the proof of Lemma \ref{lmm:head_tail_reduction} follows.

\end{proof}

\subsection{Proof of Lemma \ref{lmm:XtoY}}
\begin{proof}
Recall that the first step of converting $O$ to $O'$ was to traverse all the scheduled jobs based on their position in
$O$. To prove the first statement of Lemma \ref{lmm:XtoY}, note that all the jobs removed while traversing
 $a_{i,t}$ and $A_{i,t}^1$, have processing time $p$ and had initially intersection with interval $a_{i,t}$ with length
 $\ell_i$ in $O$. Observe that their length is $p$ and so all could be scheduled in an interval with length $\ell_i+p$. Assuming $\lambda > 3$ we have:
\begin{equation} 
x'_{i,t} + \sum_{a_{i',t'} \in A_{i,t}^1} x_{i', t'} \leq \frac{\ell_i+p}{p} \leq \frac{2\ell_i}{p} \leq 
\frac{2\ell_i}{p} \cdot \frac{3(\lambda-1)}{2\lambda} = \frac{3}{\lambda} \cdot y_{i,t}. \label{eq:XtoY_A1}
\end{equation}

To prove the second statement, observe that while traversing the jobs in $A_{i,t}^2$ we have temporarily removed $\sum_{a_{i',t'}
 \in A_{i,t}^2} x'_{i', t'}$ many jobs with processing time $p$ and they make room for the same number
 of jobs (of size $p$) in $a_{i,t}$.
Note that all $x'_{i,t}$ many jobs which are temporarily removed while traversing $a_{i,t}$ could be scheduled in 
the whole interval $Y_{i,t}$ (as their span contains $Y_{i,t}$). 
So from at most $\frac{\ell_i+p}{p}$ many jobs initially intersecting with interval
$a_{i,t}$, at most $\frac{\ell_i+p}{p} - \sum_{a_{i',t'} \in A_{i,t}^2} x_{i', t'}$ many of them would be
temporarily removed while traversing $a_{i,t}$:
\begin{equation}
x'_{i,t} + \sum_{a_{i',t'} \in A_{i,t}^2} x_{i', t'} \leq \frac{\ell_i+p}{p} \leq 
\frac{3}{\lambda} \cdot y_{i,t}. \label{eq:XtoY_A2}
\end{equation}

We only need to sum up inequalities (\ref{eq:XtoY_A1}) and (\ref{eq:XtoY_A2}) to prove Corollary \ref{cor:XtoY}:
\begin{equation*}
x'_{i,t} + \sum_{a_{i',t'} \in A_{i,t}} x'_{i', t'} \leq \big( x'_{i,t} + 
\sum_{a_{i',t'} \in A_{i,t}^1} x_{i', t'} \big) + \big( x'_{i,t} + \sum_{a_{i',t'} \in A_{i,t}^2} x_{i', t'} \big)
 \leq \frac{6}{\lambda} \cdot y_{i,t}
\end{equation*}
\end{proof}

\subsection{Proof of Lemma \ref{lmm:RtoS}}
\begin{proof}
Fix some $0 \leq i \leq k$ and $0 \leq t \leq \frac{T}{\ell_i}$. First observe that for each interval $a_{i,t}$ with 
positive $x'_{i,t}$, we have removed at most $\lceil \frac{x'_{i,t}}{\alpha_{i,t}} \rceil$ many jobs from $O$ to obtain $O''$.
This is obvious if $\alpha=1$. For $\alpha>1$ note that for each job bigger than $p$ removed from $Y_{i,t}$ we could schedule
at least $\alpha_{i,t}$ many jobs of size $p$.

\begin{equation}
\lvert R_{i,t} \rvert \leq \lceil \frac{x'_{i,t}}{\alpha_{i,t}} \rceil \leq \frac{2x'_{i,t}}{\alpha_{i,t}} \label{eq:RtoX} 
\end{equation}

Also note that the length of $Y_{i,t}$ is $(\lambda - 1) \cdot \ell_i$ and all the jobs inside $Y_{i,t}$ in solution $O''$ 
have processing time at most $p'_{i,t} \leq 2 p\alpha_{i,t}$ and between any two consecutive scheduled job there can be at most
 $p'_{i,t}$ empty space. So the time between the starting time of each two consecutive scheduled job in $Y_{i,t}$ could not
 be more than $2p'_{i,t}$.
\begin{equation}
\lvert S_{i,t} \rvert \geq \lfloor \frac{(\lambda-1) \ell_i}{4p\alpha_{i,t}} \rfloor \geq \frac{y_{i,t}}{5\alpha_{i,t}} 
\label{eq:StoY}
\end{equation}

Considering Lemma \ref{lmm:XtoY} we have: 
\begin{equation} \label{eq5}
x'_{i,t}\leq \frac{3}{\lambda} \cdot y_{i,t} 
\end{equation} 

To complete the proof of Lemma \ref{lmm:RtoS} we only need to combine Inequalities (\ref{eq:RtoX}), (\ref{eq:StoY}), and (\ref{eq5}):
\begin{equation*}
\lvert R_{i,t} \rvert \leq \frac{2x'_{i,t}}{\alpha_{i,t}} \leq \frac{2}{\alpha_{i,t}} \cdot \frac{3y_{i,t}}{\lambda} = 
\frac{30}{\lambda} \cdot \frac{y_{i,t}}{5\alpha_{i,t}} \leq \frac{30}{\lambda} \lvert S_{i,t} \rvert
\end{equation*} 
\end{proof}

\subsection{Proof of Lemma \ref{lmm:Combine_Subset_All}} 
\begin{proof} 
Fix some $0 \leq i \leq k$ and $0 \leq t \leq \frac{T}{\ell_i}$ and suppose we have sorted all intervals $a_{i', t'} 
\in A_{i,t}$ based on their $\alpha_{i', t'}$ values (in descending order) and for simplicity rename them so that
 $A_{i,t} = \{a_{i_1,t_1}, a_{i_2,t_2}, \ldots, a_{i_r,t_r}\}$ where 
$\alpha_{i_1,t_1} \geq \alpha_{i_2,t_2} \geq \ldots \geq \alpha_{i_r,t_r}$.

Suppose $h$ is the highest index where $\alpha_{i_h, t_h} \geq \alpha_{i,t}$ ($h=0$ if there is no such index).
We claim that there is an index $s$, $h \leq s \leq r$, such that the statement of Lemma \ref{lmm:Combine_Subset_All}
 holds for $A_1 = \{a_{i_1, t_1}, \ldots, a_{i_s, t_s}\}$ and $A_2 = \{a_{i_{s+1}, t_{s+1}}, \ldots, a_{i_r, t_r}\}$.
By way of contradiction suppose that the statement of Lemma \ref{lmm:Combine_Subset_All} 
is not valid for any $s$, $h \leq s \leq r$. Thus:

\begin{equation} 
\lvert R_{i,t} \cup \bigcup_{u = 1}^s R_{i_u, t_u} \rvert > \frac{60}{\lambda} \lvert S_{i,t} \setminus 
\bigcup_{u = s+1}^r S_{i_u, t_u} \rvert \label{eq:RtoS_All_S}
\end{equation} 

Also based on Lemma \ref{lmm:XtoY} we have: 

\begin{equation}
x'_{i,t} + \sum_{u = 1}^r x'_{i_u, t_u} \leq \frac{6}{\lambda} \cdot y_{i,t} = \frac{6}{\lambda} \cdot 
\frac{\lvert Y_{i, t} \rvert}{p} \label{eq:XtoY}
\end{equation} 

We are going to show that we cannot have Inequalities (\ref{eq:XtoY}) and (\ref{eq:RtoS_All_S}) for all $s$, $h \leq s \leq r$ 
at the same time and reach a contradiction. First of all to find an upper bound for the left side of Inequality 
(\ref{eq:RtoS_All_S}), observe that, by definition, for any two intervals $a_{i,t}$ and $a_{i', t'}$ there is no intersection
 between $R_{i,t}$ and $R_{i', t'}$. By using Inequality (\ref{eq:RtoX}), for each $s$, $h \leq s \leq r$ we have:
\begin{equation}
\lvert R_{i,t} \cup \bigcup_{u = 1}^s R_{i_u, t_u} \rvert = \lvert R_{i,t} \rvert + \sum_{u = 1}^s \lvert R_{i_u, t_u} 
\rvert \leq \frac{2x'_{i,t}}{\alpha_{i,t}} + \sum_{u = 1}^s \frac{2x'_{i_u,t_u}}{\alpha_{i_u,t_u}} \label{RtoX_All_S}
\end{equation}

To have a lower bound for the right side of Inequality (\ref{eq:RtoS_All_S}) we are going to define $Y_{i_u, t_u}^*$ 
for each $u$, $h < u \leq r$ and $Y_{i,t}^*$:
\begin{equation*} 
Y_{i_r, t_r}^* = Y_{i_r, t_r} \cap Y_{i,t} 
\end{equation*} 
\begin{equation*} 
h < u < r \quad \Rightarrow \quad Y_{i_u, t_u}^* = \big( Y_{i_u, t_u} \cap Y_{i,t} \big) \setminus 
\big( Y_{i_{u+1}, t_{u+1}}^* \cup \ldots \cup Y_{i_r, t_r}^* \big)
\end{equation*} 
\begin{equation*} 
Y_{i, t}^* = Y_{i,t} \setminus \big( Y_{i_{h+1}, t_{h+1}}^* \cup \ldots \cup Y_{i_r, t_r}^* \big) 
\end{equation*} 

Note that $Y_{i,t}^*$ along with all $Y_{i_u, t_u}^*$'s are a partition of $Y_{i,t}$: 
\begin{equation} 
\lvert Y_{i, t} \rvert = \lvert Y_{i, t}^* \rvert + \sum_{u = h+1}^r \lvert Y_{i_u, t_u}^* \rvert \label{eq:Y_it_partition} 
\end{equation} 

Also note that for each $u$, $h < u \leq r$ jobs scheduled inside $Y_{i_u, t_u}^*$ in $O''$ have processing time at most
 $p'_{i_u, t_u} \leq 2p \alpha_{i_u, t_u}$ and the empty space between any two consecutive scheduled job is no more than $p'_{i_u, t_u}$ too 
(otherwise we were able to add some more jobs from $X_{i,t}^p$ to $O''$), and jobs scheduled inside $Y_{i,t}^*$ have processing time at most
 $p'_{i,t} \leq 2p \alpha_{i,t}$. So for each $s$, $h \leq s \leq r$ we have:
\begin{equation} 
\lvert S_{i,t} \setminus \bigcup_{u = s+1}^r S_{i_u, t_u} \rvert
\geq \lfloor \frac{\lvert Y_{i,t}^* \rvert}{4p\alpha_{i,t}} \rfloor + \sum_{u = h+1}^s \lfloor \frac{\lvert Y_{i_u,t_u}^* \rvert}{4p\alpha_{i_u,t_u}} \rfloor
\geq \frac{\lvert Y_{i,t}^* \rvert}{5p\alpha_{i,t}} + \sum_{u = h+1}^s \frac{\lvert Y_{i_u,t_u}^* \rvert}{5p\alpha_{i_u,t_u}} \label{StoY_All_S}
\end{equation} 

The only thing we need to prove to complete the proof of the Lemma \ref{lmm:Combine_Subset_All} is that there is an
 index $s$, $h \leq s \leq r$ such that:
\begin{equation} 
\frac{x'_{i,t}}{\alpha_{i,t}} + \sum_{u = 1}^s \frac{x'_{i_u,t_u}}{\alpha_{i_u,t_u}} \leq \frac{6}{\lambda}
 \bigg( \frac{\lvert Y_{i,t}^* \rvert}{p\alpha_{i,t}} + \sum_{u = h+1}^s \frac{\lvert Y_{i_u,t_u}^* \rvert}{p\alpha_{i_u,t_u}}
 \bigg) \label{XtoY_All_S}
\end{equation} 

Combining Inequalities (\ref{RtoX_All_S}), (\ref{StoY_All_S}), and (\ref{XtoY_All_S}) completes the proof: 
\begin{equation*} 
\lvert R_{i,t} \cup \bigcup_{u = 1}^s R_{i_u, t_u} \rvert \leq \frac{2x'_{i,t}}{\alpha_{i,t}} + 
\sum_{u = 1}^s \frac{2x'_{i_u,t_u}}{\alpha_{i_u,t_u}}
\end{equation*} 
\begin{equation*} 
\leq \frac{12}{\lambda} \bigg( \frac{\lvert Y_{i,t}^* \rvert}{p\alpha_{i,t}} + 
\sum_{u = h+1}^s \frac{\lvert Y_{i_u,t_u}^* \rvert}{p\alpha_{i_u,t_u}} \bigg) \leq \frac{60}{\lambda} \lvert S_{i,t} 
\setminus \bigcup_{u = s+1}^r S_{i_u, t_u} \rvert
\end{equation*} 

Thus, we now prove Inequality (\ref{XtoY_All_S}).
We consider two cases. For the first case suppose that $h=r$, which means that $\alpha_{i,t} \leq \alpha_{i_u, t_u}$ 
for all $1 \leq u \leq r$.
Note that in this case $Y_{i,t}^* = Y_{i,t}$ and Inequality (\ref{XtoY_All_S}) would be proved using the inequality (\ref{eq:XtoY}):
\begin{equation} 
\frac{x'_{i,t}}{\alpha_{i,t}} + \sum_{u = 1}^s \frac{x'_{i_u,t_u}}{\alpha_{i_u,t_u}} \leq 
\frac{1}{\alpha_{i,t}} \big(x'_{i,t} + \sum_{u = 1}^s x'_{i_u,t_u} \big) \leq
\frac{6}{\lambda} \frac{\lvert Y_{i,t}^* \rvert}{p\alpha_{i,t}}
\end{equation} 

Hence we suppose $h < r$ and for the sake of contradiction suppose that Inequality (\ref{XtoY_All_S}) is not true
for any value of $s$. So for all $s$, $h \leq s \leq r$ we have: 
\begin{equation} 
\frac{x'_{i,t}}{\alpha_{i,t}} + \sum_{u = 1}^s \frac{x'_{i_u,t_u}}{\alpha_{i_u,t_u}} > \frac{6}{\lambda}
 \bigg( \frac{\lvert Y_{i,t}^* \rvert}{p\alpha_{i,t}} + \sum_{u = h+1}^s \frac{\lvert Y_{i_u,t_u}^* \rvert}{p\alpha_{i_u,t_u}}
 \bigg) \label{XtoY_All_S_Contradiction}
\end{equation} 

What we do is, for each value of $s$, $h\leq s\leq r$, we multiply both sides of Inequality (\ref{XtoY_All_S_Contradiction})
and sum all of them to derive a contradiction.
For $s=h$,
 multiply both sides of Inequality (\ref{XtoY_All_S_Contradiction}) 
by $\alpha_{i,t} - \alpha_{i_{h+1}, t_{h+1}}$ and for
$s = r$ multiply both sides by $\alpha_{i_r, t_r}$, and 
for every other $s$, $h < s < r$ multiply both sides of Inequality
(\ref{XtoY_All_S_Contradiction}) associated with $s$ by $\alpha_{i_s, t_s} - \alpha_{i_{s+1}, t_{s+1}}$. 
Note that considering the definition of $h$ and the fact that $h < r$, we
 have $\alpha_{i,t} > \alpha_{i_{h+1}, t_{h+1}} \geq \ldots \geq \alpha_{i_r, t_r} \geq 1$, so all the coefficients are non-negative
(and in fact the first one is positive):

\begin{equation*} 
(\alpha_{i,t} - \alpha_{i_{h+1}, t_{h+1}}) \quad \times \quad \Bigg(\frac{x'_{i,t}}{\alpha_{i,t}} + 
\frac{x'_{i_1,t_1}}{\alpha_{i_1,t_1}} + \frac{x'_{i_2,t_2}}{\alpha_{i_2,t_2}} + \ldots + \frac{x'_{i_h,t_h}}{\alpha_{i_h,t_h}} > 
\frac{6}{\lambda} \bigg( \frac{\lvert Y_{i,t}^* \rvert}{p\alpha_{i,t}} \bigg) \Bigg)
\end{equation*} 
\begin{equation*} 
(\alpha_{i_{h+1},t_{h+1}} - \alpha_{i_{h+2}, t_{h+2}}) \quad \times \quad \Bigg( \frac{x'_{i,t}}{\alpha_{i,t}} + 
\frac{x'_{i_1,t_1}}{\alpha_{i_1,t_1}} + \ldots + \frac{x'_{i_{h+1},t_{h+1}}}{\alpha_{i_{h+1},t_{h+1}}} > \frac{6}{\lambda}
 \bigg( \frac{\lvert Y_{i,t}^* \rvert}{p\alpha_{i,t}} + \frac{\lvert Y_{i_{h+1},t_{h+1}}^* \rvert}{p\alpha_{i_{h+1},t_{h+1}}}
 \bigg) \Bigg)
\end{equation*} 
\begin{equation*} 
(\alpha_{i_{h+2},t_{h+2}} - \alpha_{i_{h+3}, t_{h+3}}) \quad \times \quad \Bigg( \frac{x'_{i,t}}{\alpha_{i,t}} +
 \frac{x'_{i_1,t_1}}{\alpha_{i_1,t_1}} + \ldots + \frac{x'_{i_{h+2},t_{h+2}}}{\alpha_{i_{h+2},t_{h+2}}} >
 \frac{6}{\lambda} \bigg( \frac{\lvert Y_{i,t}^* \rvert}{p\alpha_{i,t}} +
 \frac{\lvert Y_{i_{h+1},t_{h+1}}^* \rvert}{p\alpha_{i_{h+1},t_{h+1}}} + 
\frac{\lvert Y_{i_{h+2},t_{h+2}}^* \rvert}{p\alpha_{i_{h+2},t_{h+2}}} \bigg) \Bigg)
\end{equation*} 
\begin{equation*} 
\vdots
\end{equation*} 
\begin{equation*} 
(\alpha_{i_s,t_s} - \alpha_{i_{s+1}, t_{s+1}}) \quad \times \quad \Bigg( \frac{x'_{i,t}}{\alpha_{i,t}} +
 \frac{x'_{i_1,t_1}}{\alpha_{i_1,t_1}} + \ldots + \frac{x'_{i_s,t_s}}{\alpha_{i_s,t_s}} > \frac{6}{\lambda} \bigg( \frac{\lvert Y_{i,t}^* \rvert}{p\alpha_{i,t}} + \frac{\lvert Y_{i_{h+1},t_{h+1}}^* \rvert}{p\alpha_{i_{h+1},t_{h+1}}} + \ldots + \frac{\lvert Y_{i_s,t_s}^* \rvert}{p\alpha_{i_s,t_s}} \bigg) \Bigg)
\end{equation*} 
\begin{equation*} 
\vdots
\end{equation*} 
\begin{equation*} 
(\alpha_{i_{r-1},t_{r-1}} - \alpha_{i_r, t_r}) \quad \times \quad \Bigg( \frac{x'_{i,t}}
{\alpha_{i,t}} + \frac{x'_{i_1,t_1}}{\alpha_{i_1,t_1}} + \ldots + \frac{x'_{i_{r-1},t_{r-1}}}{\alpha_{i_{r-1},t_{r-1}}} >
 \frac{6}{\lambda} \bigg( \frac{\lvert Y_{i,t}^* \rvert}{p\alpha_{i,t}} +
 \frac{\lvert Y_{i_{h+1},t_{h+1}}^* \rvert}{p\alpha_{i_{h+1},t_{h+1}}} + \ldots + 
\frac{\lvert Y_{i_{r-1},t_{r-1}}^* \rvert}{p\alpha_{i_{r-1},t_{r-1}}} \bigg) \Bigg)
\end{equation*} 
\begin{equation*} 
(\alpha_{i_r,t_r}) \quad \times \quad \Bigg( \frac{x'_{i,t}}{\alpha_{i,t}} + 
\frac{x'_{i_1,t_1}}{\alpha_{i_1,t_1}} + \ldots + \frac{x'_{i_r,t_r}}{\alpha_{i_r,t_r}} > 
\frac{6}{\lambda} \bigg( \frac{\lvert Y_{i,t}^* \rvert}{p\alpha_{i,t}} + 
\frac{\lvert Y_{i_{h+1},t_{h+1}}^* \rvert}{p\alpha_{i_{h+1},t_{h+1}}} + \ldots +
 \frac{\lvert Y_{i_r,t_r}^* \rvert}{p\alpha_{i_r,t_r}} \bigg) \Bigg)
\end{equation*} 

Now we sum up all these inequalities (with the corresponding coefficients) to reach a contradiction.
Since all coefficients are $\geq 0$ and the very first one is positive ($\alpha_{i,t} - \alpha_{i_{h+1}, t_{h+1}} > 0$)
this ensures that we have non-zero sum.
Note that for each $1 \leq s \leq h$, term $\frac{x'_{i_s,t_s}}{\alpha_{i_s,t_s}}$ has
 appeared in the left hand side of all the above 
inequalities and so its coefficient in the sum would be the sum of all the coefficients:
\begin{equation*}
(\alpha_{i,t} - \alpha_{i_{h+1}, t_{h+1}}) + (\alpha_{i_{h+1},t_{h+1}} - \alpha_{i_{h+2}, t_{h+2}}) + \ldots + 
(\alpha_{i_{r-1},t_{r-1}} - \alpha_{i_r, t_r}) + (\alpha_{i_r,t_r}) = \alpha_{i,t}
\end{equation*}

This is the case for terms $\frac{x'_{i,t}}{\alpha_{i,t}}$ and $\frac{\lvert Y_{i,t}^* \rvert}{p\alpha_{i,t}}$ as well. 
Also for each $s$, $h < s \leq r$, terms $\frac{x'_{i_s,t_s}}{\alpha_{i_s,t_s}}$ and 
$\frac{\lvert Y_{i_s,t_s}^* \rvert}{p\alpha_{i_s,t_s}}$ have appeared in the left hand side and the right hand side 
of Inequality (\ref{XtoY_All_S_Contradiction}) associated with all values
$s, s+1, s+2, \ldots, r$, respectively. So the coefficient for
 $\frac{x'_{i_s,t_s}}{\alpha_{i_s,t_s}}$ and $\frac{\lvert Y_{i_s,t_s}^* \rvert}{p\alpha_{i_s,t_s}}$ in the sum would be:
\begin{equation*}
(\alpha_{i_s,t_s} - \alpha_{i_{s+1}, t_{s+1}}) + (\alpha_{i_{h+1},t_{h+1}} - \alpha_{i_{h+2}, t_{h+2}}) + \ldots + 
(\alpha_{i_{r-1},t_{r-1}} - \alpha_{i_r, t_r}) + (\alpha_{i_r,t_r}) = \alpha_{i_s,t_s}
\end{equation*}

This means that the sum of all the inequalities written above can be simplified to:
\begin{equation*}
\alpha_{i,t} \big( \frac{x'_{i,t}}{\alpha_{i,t}} + \sum_{s = 1}^h \frac{x'_{i_s, t_s}}{\alpha_{i_s, t_s}} \big) +
 \sum_{s=h+1}^r \alpha_{i_s, t_s} \cdot \frac{x'_{i_s, t_s}}{\alpha_{i_s, t_s}} > 
\frac{6}{\lambda} \bigg( \alpha_{i,t} \frac{\lvert Y_{i,t}^* \rvert}{p\alpha_{i,t}} +
 \sum_{s=h+1}^r \alpha_{i_s, t_s} \frac{\lvert Y_{i_s,t_s}^* \rvert}{p\alpha_{i_s,t_s}} \bigg)
\end{equation*} 
\begin{equation*} 
\Longrightarrow\alpha_{i,t} \big( \frac{x'_{i,t}}{\alpha_{i,t}} + \sum_{s = 1}^h \frac{x'_{i_s, t_s}}{\alpha_{i_s, t_s}} \big) +
 \sum_{s=h+1}^r x'_{i_s, t_s} > \frac{6}{\lambda} \bigg( \frac{\lvert Y_{i,t}^* \rvert}{p} + 
\sum_{s=h+1}^r \frac{\lvert Y_{i_s,t_s}^* \rvert}{p} \bigg)
\end{equation*} 

Considering that $\alpha_{i_1,t_1} \geq \alpha_{i_2, t_2} \geq \ldots \geq \alpha_{i_h, t_h} \geq \alpha_{i,t}$, and 
Equality (\ref{eq:Y_it_partition}) we have:
\begin{equation*}
x'_{i,t} + \sum_{s=1}^r x'_{i_s, t_s} \geq \alpha_{i,t} \big( \frac{x'_{i,t}}{\alpha_{i,t}} + 
\sum_{s = 1}^h \frac{x'_{i_s, t_s}}{\alpha_{i_s, t_s}} \big) + \sum_{s=h+1}^r x'_{i_s, t_s} > 
\frac{6}{\lambda} \cdot \frac{\lvert Y_{i, t}^* \rvert + \sum_{u = h+1}^r \lvert Y_{i_u, t_u}^* \rvert}{p} = 
\frac{6}{\lambda} \cdot \frac{\lvert Y_{i, t} \rvert}{p}
\end{equation*} 
\begin{equation*} 
\Rightarrow \quad x'_{i,t} + \sum_{s=1}^r x'_{i_s, t_s} > \frac{6}{\lambda} \cdot \frac{\lvert Y_{i, t} \rvert}{p} 
\end{equation*} 

This contradicts Inequality (\ref{eq:XtoY}), which was based on Lemma \ref{lmm:XtoY} for interval $a_{i,t}$. This contradiction
show that for at least one value of $s$, Inequality (\ref{XtoY_All_S}) holds, 
 which completes the proof of Lemma \ref{lmm:Combine_Subset_All}.

\end{proof}

\section{ Proof of Theorem \ref{theo:base}}\label{sec:theobase}
In this section we prove Theorem \ref{theo:base}. We start by presenting a $(1-\varepsilon)$-approximation 
algorithm for the case of $m=1$
that runs in time $\rm{Poly}(n,p_{max})$ where $p_{max}$ is the largest processing time, and then show how to extend it to
a PTAS. 
We assume that $r_j$'s comes from a set of size $R$, $d_j$'s from a set of size $D$ where $R,D\in O(1)$. 
Also, we are given a vector $\vec{v}$ with $|\vec{v}|=B\in O(1)$ where
each $\vec{v}_i$ is a pair $(\vec{v}_i(s),\vec{v}_i(f))$ that specifies the start and end of a blocked interval over time in which
the machine cannot be used. 

Our approach will be to find windows in the time-line where jobs can feasibly be scheduled in any order; these will be windows
that do not contain any release time or deadline nor any blocked space. Each of these windows 
will be contained entirely between a pair of release times or deadlines or blocks defined by $\vec{v}$, 
so we can schedule jobs in a window in any order.
We call the pair of release time and deadline of a job its \emph{type}

\begin{definition}[Types]
We say a job $j\in J$ is of type $t=(u,v)$ if $u$ is the release time of job $j$, $r_j$, and if $v$ is the deadline 
of job $j$, $d_j$. We let $\cT$ denote the set of all job types. 
\end{definition}

Since we assume $R,D\in O(1)$, therefore $|\cT|\leq RD\in O(1)$. With these 
classifications, before scheduling individual jobs, we first guess how much processing time each job type $t$ has in an 
optimal solution and use this guess as a budget for job processing times and maximize the number of jobs of type $t$ scheduled 
given this budget. The number of such guesses will be at most $O((np_{max})^{|\cT|})\in O((np_{max})^{RD})$.

If a release time $r_j$ is within a blocked interval $(\vec{v}_i(s),\vec{v}_i(f))$ we change $r_j$ to $\vec{v}_i(f)$.
Similarly if a deadline $d_j$ is within a blocked interval $(\vec{v}_i(s),\vec{v}_i(f))$ we change $d_j$ to $\vec{v}_i(s)$.
We call the union of these release times and deadlines \emph{straddle points}, which we denote by $\cS$. 
Note that $|\cS|\leq R+D\in O(1)$.
We say a job $j$ in a schedule straddles a straddle point if it starts before the straddle point and finishes after the straddle
point (hence at the time of the straddle point the machine is busy with job $j$).

Let $\cS'$ be the union of $\vec{v}_i(s)$'s and $\vec{v}_i(f)$'s (i.e. start and end points of the blocked windows 
defined by $\vec{v}$).
For each point $\vec{v}_i(s)\in \cS'$ we assume there is a dummy job of size $\vec{v}_i(f)-\vec{v}_i(s)$ that is
being run exactly at start point $\vec{v}_i(s)$ until point $\vec{v}_i(f)$ and its position is fixed.
We enumerate the points in $\cS''=\cS\cup\cS'$ so that $s_i\in\cS''$ is the $i^{th}$ point in increasing order.

If the number of jobs in an optimum solution is smaller than $\cS''/\varepsilon=O((R+D+B)/\varepsilon)$ then we guess all these $O(1)$ 
jobs and a permutation/schedule for them in optimum and this can be done in time $n^{\cS''/\varepsilon}(\cS/\varepsilon)!$. 
So let's assume otherwise. If we remove all the
jobs in optimum that straddle a straddle point (i.e. span a release time or deadline), we incur a loss of at most $|\cS''|$ and
we are left with a solution of value at least $(1-\varepsilon)\opt$. So there is a near optimum solution with no straddle job.
Let us call such a near optimum solution $\Op$. Our goal is to find such a solution.

We define \emph{windows}, which will denote the intervals where we schedule non-straddle jobs. 
The free interval between two consecutive points in $\cS''$ define a window, i.e. the free intervals between consecutive
straddle points or between a dummy job and a straddle point. 
Let these windows be $\cW$. Note that there are at most $R+D+B$ many windows.
Before describing the algorithm, we will take the near optimal schedule $\Op$ with no straddle jobs, 
and reschedule its jobs to nicely adhere to the definitions of straddle jobs and dummy jobs and allotments (total processing
time allocated for each job type). We will also note that any feasible schedule can be left-shifted, 
meaning that the start time of any job is its release time or the end time of another job, or the start time of the interval
right after a dummy job.
This will then define \emph{canonical} schedules that we can enumerate over in our algorithm.
We will look at the schedule $\Op$ and shift-left the jobs until either:
(1) they hit their release time, or (2) hit the finish time of another job (dummy or not), or (3) hit another release time/deadline
point. Let $\vec{a}^*$ be the allotment of
jobs in each window.
Lastly, we have the following observation that will be important for finding optimal canonical schedules.
\begin{observation}\label{obs:indep_r_and_d}
Given the allotments $\vec{a}^*$, the problem of scheduling jobs of type $t$ is independent of every other job type.
\end{observation}
This last observation is important as it allows our algorithm to deal with each job type independently. This is clearly true 
since each job type has a specified allotment that jobs of that type can be scheduled in, and the allotments of two job types
 do not overlap. Given $\ell$ windows and allotments $\vec{a}_{i,t}$ for each type $t$ and window $i$ 
we have to see what is the maximum number of 
jobs of type $t$ that we can pack into these $\ell$ windows given the allotments for them in each window. This is a multiple
knapsack problem.

\subsection{Algorithm}
The algorithm here is a sweep across all canonical schedules by iterating through the windows and allotments, combined
 with a Multiple Knapsack dynamic program to schedule jobs of each type in their corresponding allotments. 
For each window $W_i\in\cW$ we guess an optimal 
choice of allotments in $W_i$, denoted $\vec{a}_i$, where $a_{i,\sigma} \in [0, np_{\max}]$ is the allotment in the $i$th 
window for jobs of type $\sigma$. We check that this choice of allotments corresponds to a canonical schedule in $W_i$ by 
checking if the allotments can be scheduled feasibly as if they were jobs (as explained below). 
More specifically,
we let window $W_1$ begin from the first straddle point $s_1$ and check that the point 
$s_1 + \sum_{\sigma=1}^{|\cT|} a_{1,\sigma}$ is at
 most I) the next straddle point or II) start of a dummy job (whichever comes first), 
if not then the check fails as the allotments are too large to fit in the window. 
We then repeat this process from start of window $W_2$ and so on.
 We also check that for any $a_{i,\sigma}\neq 0$, that the release time of type $\sigma$ is before start of window $i$,
 and the deadline of type $\sigma$ is at least end of this window, this ensures that when the jobs are scheduled in their allotments 
they are scheduled feasibly. We repeat this procedure for each window to get a choice of
allotments $\vec{a} = \{\vec{a}_i\}_{i\in [\ell]}$. If the checks succeed for each window then the allotments
can correspond to a canonical schedule.

Note that for any fixed job type, the size of an allotment
 for that type in a given window is in $[0,np_{\max}]$, so there are $O(np_{\max})$ many guesses for each job type in this window.
There are at most $(R+D+B)$ many windows and and $RD$ job types, 
so there are at most $(np_{\max})^{(R+D+B)^3}$ many allotment choices.
With a choice of allotments that correspond to a canonical schedule, 
we apply Observation \ref{obs:indep_r_and_d} to reduce the
 problem to solving an instance of the \textsc{Multiple Knapsack} problem for each job type. For the problem corresponding to 
jobs of type $\sigma$, say there is a knapsack $m_i$ corresponding to every window $i$, of size $a_{i,\sigma}$, and for each 
job $j$ of type $\sigma$ there is a corresponding item, $x_j$ in the \textsc{Multiple Knapsack} problem, with weight equal
 to $p_j$ and profit of $1$. Using a standard DP for the \textsc{Multiple Knapsack} problem with
 $R+D+B$ many knapsacks, we can solve this problem in time $O((np_{max})^{(R+B+D)^3})$.
This establishes the following lemma.

\begin{lemma}
This algorithm gives an $(1-\varepsilon)$-approximation solution to \tm with a constant number of release 
times and deadlines and blocked intervals and runs in time $(np_{max})^{(R+B+D)^3}+n^{R+D/\varepsilon}(R+D/\varepsilon)!$.
\end{lemma}

\subsection{A PTAS}
If job sizes are not assumed to be bounded by a polynomial in $n$ then the run-time of our algorithm has two problems. The first,
 is that we make $O(np_{\max})$ many guesses for each allotment. 
Second, we exactly solve the \textsc{Multiple Knapsack} problem using an algorithm with run-time that is polynomial with respect
 to both $n$ and $p_{\max}$. To deal with the second problem, we use a PTAS for the \textsc{Multiple Knapsack} problem
to find a schedule (e.g. \cite{CK05,Jansen12}).
To deal with the first problem we will use the following lemma, which states that given a $(1-\varepsilon)$-optimal 
canonical schedule $\Op$ with no straddle jobs, 
for each allotment $a_{i,t}$, if the allotment has at least $\lceil 1/\varepsilon^2 \rceil$ jobs then we can reduce the size of 
the allotment to the nearest power of $(1 + \varepsilon)$ and drop jobs in order from largest to smallest until the remaining 
jobs can be scheduled entirely in this reduced allotment, at a loss of factor at most $1-2\varepsilon$.

\begin{lemma} \label{lemma:ptas}
Given a canonical schedule $\Op$, if we apply the above rounding procedure then the throughput of this new schedule is a 
$(1 - 2\varepsilon)$-approximation of the throughput of $\Op$.
\end{lemma}
\begin{proof}
Take a canonical schedule $\Op$. For a fixed window, if an allotment has at least $\alpha = \lceil 1/\varepsilon^2 \rceil$
 jobs then we round down the size of the allotment to the nearest power of $(1 + \varepsilon)$. We drop jobs in order of 
largest to smallest until the remaining jobs fit in the allotment.

We want to show that the fraction of jobs remaining after this rounding is at least $\frac{1}{1+\varepsilon}$. The worst case
 for this fraction is when the jobs in this allotment is exactly $\alpha$ many jobs. Rounding the allotment size down to the 
nearest $(1+\varepsilon)$ power means that there will be at least $\lfloor\frac{\alpha}{1+\varepsilon}\rfloor$ jobs. If we let 
$\alpha=\frac{1}{\varepsilon^2}$, the fraction of jobs remaining will be at least $(1-2\varepsilon)$.
\end{proof}

So the number of guesses we have to make for allotment of each job type in each window will reduce from $O(np_{max})$ to
$O(\log(np_{max}))$. The algorithm we use will be similar to the pseudo-polynomial time algorithm. 
We will sweep across the windows as before, checking that they 
correspond to canonical schedules. To sweep across allotments, we will guess from both allotment sizes that are powers of 
$(1+\varepsilon)$ and that are equal to combinations of up to $\lceil 1/\varepsilon^2 \rceil$ many job sizes. 
This reduces the number of guesses from $(np_{max})^{(R+D+B)^3}$ to $(\log(np_{max}))^{(R+D+B)^3}$. 
The reduction to the \textsc{Multiple Knapsack} problem is the same but instead of the pseudo-polynomial time 
solution, we use the PTAS due to \cite{Jansen12} which runs in time $2^{\eps^{-1}\log^{-4}(1/\eps)}+\rm{Poly}(n)$. 
The proof of the following is immediate.

\begin{lemma}
This algorithm runs in polynomial time.
\end{lemma}

\begin{theorem}
This algorithm is a PTAS for the \tm problem with a constant number of release times and deadlines and blocked intervals.
\end{theorem}
\begin{proof}
We know we restrict our choices of allotments to be either the case that the size of the allotment is some rounded value, or
 that are combinations of up $\lceil 1/\varepsilon^2 \rceil$ many jobs. As we have shown in Lemma \ref{lemma:ptas} this will
 give an allotment whose optimal packing is within $1-2\varepsilon$ of the optimal value for that job type and window.

Given this choice of allotments, a solution to \textsc{Multiple Knapsack} problems with constant many knapsacks
with unit weighted jobs of arbitrary size can be solved using a PTAS due to \cite{Jansen12}.
Therefore, we find a solution that is at least a $(1-\varepsilon)(1-2\varepsilon) (1-\varepsilon) = 1 - O(\varepsilon)$ factor 
of the optimal solution where one $1-\varepsilon$ factor is to assume there are no straddle jobs, on $1-\varepsilon$ factor is
due to use of a PTAS for the \textsc{Multiple Knapsack} problem, and the $1-2\varepsilon$ factor is due to the rounding up
the guessed sizes of allotments to powers of $1+\varepsilon$. Total time will be 
$(2^{\eps^{-1}\log^{-4}(1/\eps)}+{\rm Poly}(n))(\log(np_{max}))^{(R+D+B)^3}=O( 2^{\eps^{-1}\log^{-4}(1/\eps)}+{\rm Poly}(n))$.
\end{proof}

\subsection{Extending to a Constant Number of Machines}
In this subsection we describe how to extend the results of this section to a constant number of machines. We first describe
the extension of the pseudo-polynomial time algorithm. The intuition of this extension is simple,
 as before we assume there are no straddle jobs at a loss of $1-\varepsilon$ factor. Let $\Op$ be a 
 $(1-\varepsilon)$-approximate solution with no job straddling a straddle point. Windows are defined similarly.
We guess allotments for each job type, for each window and for each machine.
The number of windows increases by at most a factor of $m$ so the number of possible allotment guesses is bounded by
 $(np_{\max})^{m(R+D+B)^3}$. 
With multiple machines we can define \emph{canonical} schedules in a similar way as the single machine case.
The algorithm is a straightforward extension of the algorithm for single machine. We guess the allotments $\vec{a}$ for 
all the windows as in $\Op$. To check these choices correspond to a canonical schedule, we perform the check 
described earlier on a machine by machine basis. To find the schedule given these allotments, we perform the 
same reduction to the \textsc{Multiple Knapsack} problem. Since the 
number of knapsacks increases by a factor of at most $m=O(1)$, the algorithm still runs in time polynomial in $n$ and $p_{\max}$.

We also have that Lemma \ref{lemma:ptas} holds for this problem since it argues on a per allotment basis. So we can get a PTAS
 for this problem by guessing allotments that are either powers or
$(1+\varepsilon)$ or are equal to combinations of up to $\lceil 1/\varepsilon^2 \rceil$ many job sizes. We reduce to the
$\textsc{Multiple Knapsack}$ problem as before and again apply the PTAS due to \cite{Jansen12},
noting that since the number of allotments increase by a factor of at most $m$, the algorithm of \cite{Jansen12}
still runs in polynomial time.

\bibliographystyle{plainurl}
{
\bibliography{ref-throughput}

\begin{thebibliography}{10}

\bibitem{ARSU02}
Micah Adler, Arnold~L. Rosenberg, Ramesh~K. Sitaraman, and Walter Unger.
\newblock Scheduling time-constrained communication in linear networks.
\newblock {\em Theory Comput. Syst.}, 35(6):599--623, 2002.
\newblock URL: \url{https://doi.org/10.1007/s00224-002-1001-6}.

\bibitem{BCKPSS07}
Nikhil Bansal, Ho{-}Leung Chan, Rohit Khandekar, Kirk Pruhs, Clifford Stein,
  and Baruch Schieber.
\newblock Non-preemptive min-sum scheduling with resource augmentation.
\newblock In {\em 48th Annual {IEEE} Symposium on Foundations of Computer
  Science {(FOCS} 2007), October 20-23, 2007, Providence, RI, USA,
  Proceedings}, pages 614--624, 2007.
\newblock URL: \url{https://doi.org/10.1109/FOCS.2007.46}.

\bibitem{B03}
Philippe Baptiste.
\newblock On minimizing the weighted number of late jobs in unit execution time
  open-shops.
\newblock {\em European Journal of Operational Research}, 149(2):344--354,
  2003.
\newblock URL: \url{https://doi.org/10.1016/S0377-2217(02)00759-2}.

\bibitem{BBKT}
Philippe Baptiste, Peter Brucker, Sigrid Knust, and Vadim~G. Timkovsky.
\newblock Ten notes on equal-processing-time scheduling.
\newblock {\em 4OR}, 2(2):111--127, 2004.
\newblock URL: \url{https://doi.org/10.1007/s10288-003-0024-4}.

\bibitem{BBFNS01}
Amotz Bar{-}Noy, Reuven Bar{-}Yehuda, Ari Freund, Joseph Naor, and Baruch
  Schieber.
\newblock A unified approach to approximating resource allocation and
  scheduling.
\newblock {\em J. {ACM}}, 48(5):1069--1090, 2001.
\newblock URL: \url{https://doi.org/10.1145/502102.502107}.

\bibitem{BGNS01}
Amotz Bar{-}Noy, Sudipto Guha, Joseph Naor, and Baruch Schieber.
\newblock Approximating the throughput of multiple machines in real-time
  scheduling.
\newblock {\em {SIAM} J. Comput.}, 31(2):331--352, 2001.
\newblock URL: \url{https://doi.org/10.1137/S0097539799354138}.

\bibitem{BKMMRRSW92}
Sanjoy~K. Baruah, Gilad Koren, Decao Mao, Bhubaneswar Mishra, Arvind
  Raghunathan, Louis~E. Rosier, Dennis~E. Shasha, and Fuxing Wang.
\newblock On the competitiveness of on-line real-time task scheduling.
\newblock {\em Real-Time Systems}, 4(2):125--144, 1992.
\newblock URL: \url{https://doi.org/10.1007/BF00365406}.

\bibitem{BD00}
Piotr Berman and Bhaskar DasGupta.
\newblock Improvements in throughout maximization for real-time scheduling.
\newblock In {\em Proceedings of the Thirty-Second Annual {ACM} Symposium on
  Theory of Computing, May 21-23, 2000, Portland, OR, {USA}}, pages 680--687,
  2000.
\newblock URL: \url{https://doi.org/10.1145/335305.335401}.

\bibitem{CK05}
Chandra Chekuri and Sanjeev Khanna.
\newblock A polynomial time approximation scheme for the multiple knapsack
  problem.
\newblock {\em {SIAM} J. Comput.}, 35(3):713--728, 2005.
\newblock URL: \url{https://doi.org/10.1137/S0097539700382820}.

\bibitem{CGKN04}
Julia Chuzhoy, Sudipto Guha, Sanjeev Khanna, and Joseph Naor.
\newblock Machine minimization for scheduling jobs with interval constraints.
\newblock In {\em 45th Symposium on Foundations of Computer Science {(FOCS}
  2004), 17-19 October 2004, Rome, Italy, Proceedings}, pages 81--90, 2004.
\newblock URL: \url{https://doi.org/10.1109/FOCS.2004.38}.

\bibitem{CN06}
Julia Chuzhoy and Joseph Naor.
\newblock New hardness results for congestion minimization and machine
  scheduling.
\newblock {\em J. {ACM}}, 53(5):707--721, 2006.
\newblock URL: \url{https://doi.org/10.1145/1183907.1183908}.

\bibitem{COR06}
Julia Chuzhoy, Rafail Ostrovsky, and Yuval Rabani.
\newblock Approximation algorithms for the job interval selection problem and
  related scheduling problems.
\newblock {\em Math. Oper. Res.}, 31(4):730--738, 2006.
\newblock URL: \url{https://doi.org/10.1287/moor.1060.0218}.

\bibitem{RS09}
Mitre Dourado, Rosiane Rodrigues, and Jayme Szwarcfiter.
\newblock Scheduling unit time jobs with integer release dates to minimize the
  weighted number of tardy jobs.
\newblock {\em Annals of Operations Research}, 169(1):81--91, 2009.
\newblock URL:
  \url{https://EconPapers.repec.org/RePEc:spr:annopr:v:169:y:2009:i:1:p:81-91:10.1007/s10479-008-0479-y}.

\bibitem{EW14}
Jan Elffers and Mathijs de~Weerdt.
\newblock Scheduling with two non-unit task lengths is np-complete.
\newblock {\em CoRR}, abs/1412.3095, 2014.
\newblock URL: \url{http://arxiv.org/abs/1412.3095}, \href
  {http://arxiv.org/abs/1412.3095} {\path{arXiv:1412.3095}}.

\bibitem{FN95}
Ulrich Faigle and Willem~M. Nawijn.
\newblock Note on scheduling intervals on-line.
\newblock {\em Discrete Applied Mathematics}, 58(1):13--17, 1995.
\newblock URL: \url{https://doi.org/10.1016/0166-218X(95)00112-5}.

\bibitem{FMT89}
Matteo Fischetti, Silvano Martello, and Paolo Toth.
\newblock The fixed job schedule problem with working-time constraints.
\newblock {\em Operations Research}, 37(3):395--403, 1989.
\newblock URL: \url{https://doi.org/10.1287/opre.37.3.395}.

\bibitem{GJ79}
M.~R. Garey and David~S. Johnson.
\newblock {\em Computers and Intractability: {A} Guide to the Theory of
  NP-Completeness}.
\newblock W. H. Freeman, 1979.

\bibitem{G04}
Martin~Charles Golumbic.
\newblock {\em Algorithmic Graph Theory and Perfect Graphs}.
\newblock North-Holland Publishing Co. Amsterdam, The Netherlands, 2004.

\bibitem{HR98}
Roshdy H.~M. Hafez and G.~R. Rajugopal.
\newblock Adaptive rate controlled, robust video communication over packet
  wireless networks.
\newblock {\em {MONET}}, 3(1):33--47, 1998.
\newblock URL: \url{https://doi.org/10.1023/A:1019156211458}.

\bibitem{ILM17}
Sungjin Im, Shi Li, and Benjamin Moseley.
\newblock Breaking 1 - 1/e barrier for non-preemptive throughput maximization.
\newblock In {\em Integer Programming and Combinatorial Optimization - 19th
  International Conference, {IPCO} 2017, Waterloo, ON, Canada, June 26-28,
  2017, Proceedings}, pages 292--304, 2017.
\newblock URL: \url{https://doi.org/10.1007/978-3-319-59250-3\_24}.

\bibitem{ILMT15}
Sungjin Im, Shi Li, Benjamin Moseley, and Eric Torng.
\newblock A dynamic programming framework for non-preemptive scheduling
  problems on multiple machines: Extended abstract.
\newblock In {\em Proceedings of the Twenty-Sixth Annual ACM-SIAM Symposium on
  Discrete Algorithms}, SODA 15, pages 1070--1086, USA, 2015. Society for
  Industrial and Applied Mathematics.

\bibitem{Jansen12}
Klaus Jansen.
\newblock A fast approximation scheme for the multiple knapsack problem.
\newblock In {\em {SOFSEM} 2012: Theory and Practice of Computer Science - 38th
  Conference on Current Trends in Theory and Practice of Computer Science,
  {\v{S}}pindler{\r{u}}v Ml{\'{y}}n, Czech Republic, January 21-27, 2012.
  Proceedings}, pages 313--324, 2012.
\newblock URL: \url{https://doi.org/10.1007/978-3-642-27660-6\_26}.

\bibitem{KS92}
Gilad Koren and Dennis~E. Shasha.
\newblock D\({}^{\mbox{over}}\); an optimal on-line scheduling algorithm for
  overloaded real-time systems.
\newblock In {\em Proceedings of the Real-Time Systems Symposium - 1992,
  Phoenix, Arizona, USA, December 1992}, pages 290--299, 1992.
\newblock URL: \url{https://doi.org/10.1109/REAL.1992.242650}.

\bibitem{LT94}
Richard~J. Lipton and Andrew Tomkins.
\newblock Online interval scheduling.
\newblock In {\em Proceedings of the Fifth Annual {ACM-SIAM} Symposium on
  Discrete Algorithms. 23-25 January 1994, Arlington, Virginia, {USA.}}, pages
  302--311, 1994.
\newblock URL: \url{http://dl.acm.org/citation.cfm?id=314464.314506}.

\bibitem{LZ98}
Hang Liu and Magda~El Zarki.
\newblock Adaptive source rate control for real-time wireless video
  transmission.
\newblock {\em {MONET}}, 3(1):49--60, 1998.
\newblock URL: \url{https://doi.org/10.1023/A:1019108328296}.

\bibitem{Potts}
Chris~N. Potts and Vitaly~A. Strusevich.
\newblock Fifty years of scheduling: a survey of milestones.
\newblock {\em {JORS}}, 60({S1}), 2009.
\newblock URL: \url{https://doi.org/10.1057/jors.2009.2}.

\bibitem{PST04}
Kirk Pruhs, Jir{\'{\i}} Sgall, and Eric Torng.
\newblock Online scheduling.
\newblock In {\em Handbook of Scheduling - Algorithms, Models, and Performance
  Analysis.} 2004.
\newblock URL:
  \url{http://www.crcnetbase.com/doi/abs/10.1201/9780203489802.ch15}.

\bibitem{RT87}
Prabhakar Raghavan and Clark~D. Thompson.
\newblock Randomized rounding: a technique for provably good algorithms and
  algorithmic proofs.
\newblock {\em Combinatorica}, 7(4):365--374, 1987.
\newblock URL: \url{https://doi.org/10.1007/BF02579324}.

\bibitem{SW99}
Petra Schuurman and Gerhard~J. Woeginger.
\newblock Polynomial time approximation algorithms for machine scheduling: ten
  open problems.
\newblock {\em Journal of Scheduling}, 2(5):203--213, 1999.

\bibitem{S12}
Jir{\'{\i}} Sgall.
\newblock Open problems in throughput scheduling.
\newblock In {\em Algorithms - {ESA} 2012 - 20th Annual European Symposium,
  Ljubljana, Slovenia, September 10-12, 2012. Proceedings}, pages 2--11, 2012.
\newblock URL: \url{https://doi.org/10.1007/978-3-642-33090-2\_2}.

\bibitem{S98}
Frits C.~R. Spieksma.
\newblock Approximating an interval scheduling problem.
\newblock In {\em Approximation Algorithms for Combinatorial Optimization,
  International Workshop APPROX'98, Aalborg, Denmark, July 18-19, 1998,
  Proceedings}, pages 169--180, 1998.
\newblock URL: \url{https://doi.org/10.1007/BFb0053973}.

\bibitem{YL97}
David K.~Y. Yau and Simon~S. Lam.
\newblock Adaptive rate-controlled scheduling for multimedia applications.
\newblock {\em {IEEE/ACM} Trans. Netw.}, 5(4):475--488, 1997.
\newblock URL: \url{https://doi.org/10.1109/90.649461}.

\end{thebibliography}
}

\end{document}